\documentclass[a4paper,UKenglish,cleveref,autoref, thm-restate]{lipics-v2021}
\usepackage{graphicx}
\title{Reconfiguring Shortest Paths in Graphs\footnote{This project has received funding from the European Union’s Horizon 2020 research and innovation programmed under grant agreement No. 682203-ERC-[Inf-Speed-Tradeoff].}} 

\titlerunning{Reconfiguring Shortest Paths in Graphs} 

\author{Kshitij Gajjar}{National University of Singapore, Singapore}{kshitijgajjar@gmail.com}{}{Research supported by NUS ODPRT Grant WBS No. R-252-000-A94-133.}

\author{Agastya Vibhuti Jha}{\'Ecole polytechnique f\'ed\'erale de Lausanne, Switzerland}{agastya.jha@epfl.ch}{}{}

\author{Manish Kumar}{Ben-Gurion University of the Negev, Israel}{manishk@post.bgu.ac.il}{}{}

\author{Abhiruk Lahiri}{Ariel University, Israel}{abhiruk@ariel.ac.il}{}{Research supported by Ariel University Post-doctoral fellowship, Israel Science Foundation, grant number 592/17 and 822/18 and ERC-CZ project LL2005.}

\authorrunning{K. Gajjar, A. V. Jha, M. Kumar and A. Lahiri} 

\Copyright{Kshitij Gajjar and Agastya Vibhuti Jha and Manish Kumar and Abhiruk Lahiri} 



\acknowledgements{Agastya Vibhuti Jha would like to thank Dr. Jatin Batra for introducing him to Ordered Optimization, which led to the idea of $k$-SPR.}

\nolinenumbers 

\hideLIPIcs  

\EventEditors{John Q. Open and Joan R. Access}
\EventNoEds{2}
\EventLongTitle{42nd Conference on Very Important Topics (CVIT 2016)}
\EventShortTitle{CVIT 2016}
\EventAcronym{CVIT}
\EventYear{2016}
\EventDate{December 24--27, 2016}
\EventLocation{Little Whinging, United Kingdom}
\EventLogo{}
\SeriesVolume{42}
\ArticleNo{23}

\usepackage{algorithm,algpseudocode}
\usepackage{xcolor,hyperref}
\usepackage{amsmath,amsthm,amsfonts,amssymb,mathtools}
\usepackage{complexity}
\usepackage{booktabs}

\newtheorem{fact}{Fact}
\DeclareMathOperator{\spr}{\textsc{SPR}}
\newcommand{\kspr}{k\textsc{-SPR}}
\DeclareMathOperator{\bfs}{BFS}
\DeclareMathOperator{\look}{Lookup}
\DeclareMathOperator{\reduc}{\mathsf{REDUC}}
\DeclareMathOperator{\diam}{\mathsf{diam}}
\DeclareMathOperator{\alg}{ALG}
\DeclareMathOperator{\opt}{OPT}
\DeclareMathOperator{\xor}{\textsf{XOR}}
\DeclareMathOperator{\ttt}{TT}
\DeclareMathOperator{\tb}{TB}
\DeclareMathOperator{\bt}{BT}
\DeclareMathOperator{\bb}{BB}
\DeclareMathOperator{\ssb}{sB}
\DeclareMathOperator{\tst}{Tt}
\DeclareMathOperator{\bst}{Bt}

\newcommand{\ones}{\mathsf{ones}}

\hypersetup{
	colorlinks,
	linkcolor = red!90!white!60!black,
	citecolor = blue!90!white!60!black,
	urlcolor  = green!90!white!60!black
}

\begin{document}
\maketitle
\begin{abstract}
    Reconfiguring two shortest paths in a graph means modifying one shortest path to the other by changing one vertex at a time, so that all the intermediate paths are also shortest paths. This problem has several natural applications, namely: (a) revamping road networks, (b) rerouting data packets in a synchronous multiprocessing setting, (c) the shipping container stowage problem, and (d) the train marshalling problem.
    
    When modelled as graph problems, (a) is the most general case while (b), (c) and (d) are restrictions to different graph classes. We show that (a) is intractable, even for relaxed variants of the problem. For (b), (c) and (d), we present efficient algorithms to solve the respective problems. We also generalize the problem to when at most $k$ (for a fixed integer $k\geq 2$) contiguous vertices on a shortest path can be changed at a time.
\end{abstract}

\newpage

\tableofcontents

\newpage

\section{Introduction}
A~\emph{reconfiguration problem} asks computational questions of the following kind: Given two different configurations of a system, is it possible to gradually transform one to the other? Two popular examples of reconfiguration problems are the 15-puzzle~\cite{RatnerW86, Goldreich11} and the Rubik's cube~\cite{DemaineDELW11, DemaineER18}. In both, we want to determine how to reach a ``solved'' final configuration using a sequence of ``moves'', starting from a given initial configuration. Recently, a lot of research has gone into the study of different types of reconfiguration problems on graphs~\cite{MouawadNPR17, LokshtanovM18, LokshtanovMPRS18, MouawadNRS18}.

In this paper, we undertake a theoretical study of the reconfiguration problem on shortest paths in graphs, known as the~\emph{Shortest Path Reconfiguration} problem (abbreviated as $\spr$), introduced by~\cite{KaminskiMM10}. Let us now define this formally.

\begin{definition}
Given an undirected, unweighted graph $G$ with a source vertex $s$ and a target vertex $t$, we say that two $s$--$t$ shortest paths $P$ and $Q$ in $G$ are~\emph{reconfigurable} if there is a sequence of $s$--$t$ shortest paths $(P_0, P_1,\ldots,P_{k-1}, P_k)$ where $P_0=P$ and $P_k=Q$ (for some positive integer $k$) such that $P_i$ and $P_{i+1}$ (for each $i\in\{0,1,\ldots,k-1\}$) differ in only one vertex. (See~\autoref{fig:Reconf} for an example.)
\end{definition}

\begin{figure}[h]
    \centering
    \vspace{0.5cm}
    \includegraphics[width=0.9\linewidth]{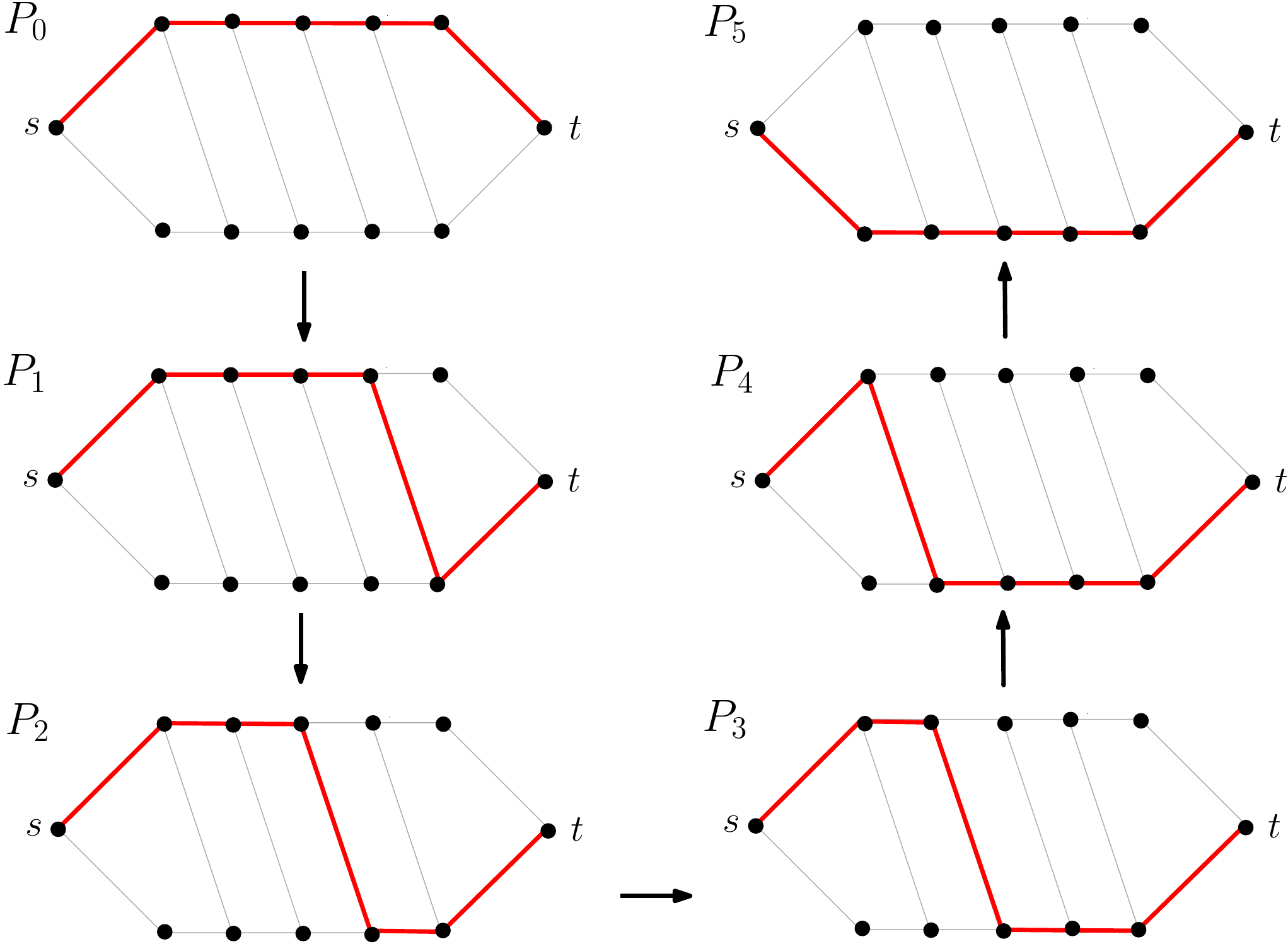}
    \vspace{0.5cm}
    \caption{Reconfiguring a path $P=P_0$ to another path $Q=P_5$ by changing one vertex at a time. Note that all paths in the sequence $(P_0,P_1,P_2,P_3,P_4,P_5)$ are $s$--$t$ shortest paths.}
    \label{fig:Reconf}
\end{figure}

$\spr$ is the decision problem of checking whether two given shortest paths in a graph are reconfigurable. Additionally, one may ask questions of the following form: If two shortest paths are indeed reconfigurable, is the reconfiguration sequence short enough? If so, is the sequence efficiently computable?

$\spr$ has several real-world applications, some of which we describe in~\autoref{sec:applications}. Despite these numerous applications, $\spr$ has not received its fair share of attention from the theoretical standpoint. This is because when research on reconfiguration began almost forty years ago, the main motivation behind studying the problem was in the context of coordinated motion planning of robots~\cite{HopcroftSS}. Large swarms of robots are operated by a central algorithm, which gives specific instructions to each robot so that they can function as a team to solve a given task. Given the initial and final configurations of the robots (a configuration is simply a snapshot of the positions of the robots), the goal of the algorithm is to modify the initial positions of the robots in a sequential, step-by-step manner so that they can reach the final configuration without bumping into each other. (A possible usage of the robots in this setting is to manage a warehouse or inventory.)

Soon thereafter, it was shown that coordinated motion planning of robots is $\PSPACE$-complete~\cite{HopcroftW86}, implying that there is no polynomial-time algorithm for it unless $\P=\PSPACE$. Another closely related problem that was studied roughly around the same time was known as~\emph{2-dimensional planar linkage}~\cite{HopcroftJW84}. Although it was not explicitly stated, it is easy to observe that 2-dimensional planar linkage is essentially a problem about reconfiguring paths of a fixed length on a graph, which is a special case of $\spr$. For two decades after that, this observation went unexplored and theoretical research in $\spr$ remained dormant. Recently however, there has been a flurry of papers on $\spr$~\cite{KaminskiMM10, Bonsma13, Bonsma17,AsplundEHHNW18, AsplundW20}.

\cite{Bonsma13} showed that $\spr$ is $\PSPACE$-complete in general. A careful look at their proof further tells us the following.
\begin{observation}
\label{obs:bipart}
$\spr$ is $\PSPACE$-hard even if the input graphs are restricted to be bipartite.
\end{observation}
(For completeness, we provide a proof of~\autoref{obs:bipart} in~\autoref{sec:bdd-diam}.) On the positive side, it known that $\spr$ is solvable in polynomial time for certain graph classes such as planar graphs~\cite{Bonsma17}, grid graphs~\cite{AsplundEHHNW18}, claw-free graphs and chordal graphs~\cite{Bonsma13}.

In this paper, we further investigate the complexity of $\spr$, particularly focusing on graph classes that model real-world problems.

\section{Our Contributions and Paper Roadmap}

Our contributions are twofold. First, we study $\spr$ for various graph classes. And second, we introduce a generalized variant of $\spr$ called $\kspr$. Alongside, we provide a roadmap of our paper (\autoref{tab:SPR}).

\begin{table}[ht]
\centering
\vspace{0.5cm}
\begin{tabular}{|l|l|l|}
\hline
\textbf{Graph Class} & \textbf{Application (Subsection)} & \textbf{Result (Subsection)}  \\ 
\hline
General graphs ($\kspr$) & Revamping Road Networks (\ref{sec:revampingroadnetworks}) & $\PSPACE$-complete (\ref{sec:k-spr-line}) \\
Permutation graphs & Train Marshalling (\ref{sec:trainmarshalling}) & Polynomial time (\ref{sec:permute}) \\
Circle graphs & Shipping Conatiner Stowage (\ref{sec:shippingcontainerstowage}) & Polynomial time (\ref{sec:circle}) \\
Bridged graphs & - & Polynomial time (\ref{sec:bridged}) \\
Boolean hypercube & Rerouting Data Packets (\ref{sec:datapacketrerouting}) & Polynomial time (\ref{sec:boolean}) \\
Circular-arc graphs & - & Polynomial time (\ref{sec:circarc}) \\
Constant diameter graphs & - & Polynomial time (\ref{sec:bdd-diam}) \\
Line graphs ($\kspr$) & - & $\PSPACE$-complete (\ref{sec:k-spr-line}) \\
Graph powers & - & $\PSPACE$-complete (\ref{sec:k-spr-power}) \\
\hline
\end{tabular}
\vspace{0.5cm}
\caption{Roadmap of our paper}
\label{tab:SPR}
\end{table}

\subsection{\texorpdfstring{$\spr$}{SPR}: Boolean Hypercube, Circle Graphs, and More}

For circle graphs, permutation graphs and the Boolean hypercube, we provide a complete characterisation of shortest paths and their reconfigurability for $\spr$. This automatically yields polynomial-time algorithms for them.

For the Boolean hypercube, we show that every shortest path corresponds to a permutation. In fact, the length of the shortest reconfiguration sequence between two shortest paths is precisely the~\emph{Kendall's Tau distance}~\cite{Sedgewick} between their respective permutations. 
The characterisation for circle graphs and permutation graphs is slightly more technically involved. 
We also solve $\spr$ in polynomial time for a subclass of metric graphs called bridged graphs (more generally, for weakly modular graphs), using a dynamic programming algorithm. 
Finally, for circular-arc graphs and graphs of bounded diameter, we observe that $\spr$ admits simple polynomial-time algorithms. 

\subsection{\texorpdfstring{$\kspr$}{k-SPR}: Hardness and Optimization Variants}

We introduce a novel generalisation of $\spr$ called $\kspr$, in which we are allowed to change at most $k$ successive vertices (instead of only one vertex) at a time.

It is known that $\spr$ can be solved in polynomial time for line graphs~\cite{Bonsma13}. We show that $\kspr$ is $\PSPACE$-complete for line graphs for all integer constants $k\geq 2$, demonstrating that $\kspr$ can be significantly harder than $\spr$. 
We also use a ``lift-and-project'' type proof to show that $\spr$ is $\PSPACE$-complete for graph powers by using the $\PSPACE$-hardness of $\kspr$. 

We observe that for a fixed $n$, the computational complexity of $\kspr$ can decrease as $k$ increases. More precisely, $\kspr$ can be solved in polynomial time when $k\geq n/2$ (\autoref{sec:gradationcomplexitykspr}). We also study a few optimisation variants of $\kspr$, and show that there is no polynomial-time algorithm to approximate the ``cost'' of $\kspr$ within a factor of $O(2^{n^2})$, unless $\PSPACE=\P$ (\autoref{sec:optvar}). Finally, we examine the gradation of the maximum number of different shortest paths in $n$-vertex graphs as the distance between $s$ and $t$ varies from $O(1)$ to $\Omega(n)$ (\autoref{sec:tradeoff}).

\section{Applications} \label{sec:applications}

In this section, we study the $\spr$ problem on some restricted graph classes, and discuss their usage in practice.

\subsection{The Shipping Container Stowage Problem} \label{sec:shippingcontainerstowage}

\begin{figure}
    \centering
    \vspace{0.5cm}
    \includegraphics[width=0.5\linewidth]{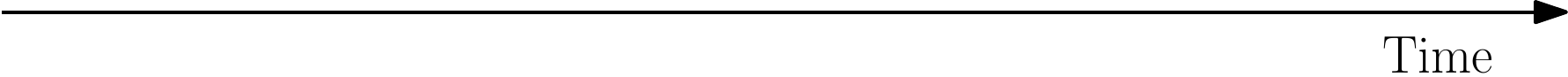}
    \includegraphics[width=0.5\linewidth,
    page=3]{Fig_2.pdf}
    \vspace{0.5cm}
    \caption{An overlap graph (above) for a cargo ship, and its corresponding circle graph (below). Each interval in the overlap graph represents a cargo container, its two end points being the loading/unloading times of the container on the ship. \textit{Figure inspired by a similar figure from~\cite{Gavril73}.}}
    \label{fig:interval}
\end{figure}

$\spr$ for circle graphs is applicable in maritime transport. Around $80\%$ of all traded goods are transported by sea~\cite{Stat1}. Cargo shipping is a billion dollar industry which leaves a considerable carbon footprint on the environment~\cite{Stat2}. Therefore, an efficient process for stowing freight containers on cargo ships is desirable. The process of shifting these containers is an expensive, time-consuming and delicate task. The problem of minimizing the amount of shifting, given a ship's voyage plan, is known as the~\emph{container stowage problem}. Owing to its importance, this problem has been studied extensively~\cite{WilsonR00, AvrielPSW98, AvrielPS00, TierneyPJ14,gajrad}. A slight variation of this problem, called the~\emph{blocks relocation problem} has also been studied~\cite{CasertaVS11, CasertaSV12}.

One can model the container stowage problem as a graph by representing each container as a vertex, wherein two vertices are adjacent if and only if loading one container necessitates unloading the other. These graphs are called~\emph{overlap graphs}. In fact, a graph is an overlap graph if and only if it is a~\emph{circle graph}~\cite{Gavril73} (\autoref{fig:interval}). Using this, it was shown that it is $\NP$-complete to minimize the amount of unloading/reloading of containers~\cite{AvrielPS00, TierneyPJ14}. However, there are two heuristics that give an approximate solution efficiently~\cite{WilsonR00, CasertaSV12}. One heuristic uses a shortest path-based solution~\cite{CasertaSV12}, while the other reshuffles the containers in a smart way while limiting the number of possible moves for each container~\cite{AvrielPSW98}.

When the containers are reshuffled at a port, a major operational challenge is to maintain the quality of the solution. Unloading a container (say $C$) at its destination port requires removing the containers stowed above it (called overstowed containers). As all these overstowed containers are adjacent to the vertex $C$ in the overlap graph, a good strategy is to maintain a shortest path from $C$ to the vertex that corresponds to the container at the top of $C$'s stack at each port. If an extra container is added at some port or an existing container is removed from some port, we should be able to quickly reconfigure the earlier unloading/reloading configuration to a new optimal unloading/reloading configuration.

\subsection{The Train Marshalling Problem} \label{sec:trainmarshalling}

We solve $\spr$ for circle graphs by solving $\spr$ for a subclass of circle graphs known as~\emph{permutation graphs}, and then generalizing our solution to circle graphs. In fact, permutation graphs themselves model a problem that is very much similar to container stowage called the train marshalling problem~\cite{DahlhausHMR00, JaehnRW15, RinaldiR17, DorpinghausS18, FalsafainT20}.

Both permutation graphs and circle graphs also have applications in memory allocation for system programs~\cite{EvenI, EvenPL72}. For a comprehensive survey on permutation graphs and circle graphs, see~\cite{Golumbic, Brandstadt}.

\subsection{Rerouting Data Packets} \label{sec:datapacketrerouting}

In an efficient synchronous multiprocessing environment, it is widely assumed that there is a common memory and processors having sequential capabilities can access it simultaneously and almost arbitrarily~\cite{ValiantB81}. Such a network of processors is realised by a $d$-\emph{dimensional Boolean hypercube}~\cite{HayesMSCP86}. The routing of message packets in such a network happens via a greedy scheme which follows shortest paths~\cite{StamoulisT94}. The main challenge here is to perform routing in a congestion-free manner, and a lot of research had gone into this~\cite{PifarreGFS94, GrammatikakisHS98}. A natural solution is to gradually reroute the packets to a different route~\cite{GreenbergH92}, which is precisely the $\spr$ problem on the Boolean hypercube.

\subsection{Revamping Road Networks} \label{sec:revampingroadnetworks}

$\kspr$ has a natural application in restructuring road networks. Suppose you are a city planner and your city's road network needs to be revamped to better serve the requirements of its residents. For this, you want to change the route between two point locations $s$ and $t$ in the city. It is not possible to change the entire route in one go, as laying out new roads takes resources, effort and time. Furthermore, this transition should be smooth. You do not want your ongoing renovation project to cause undue congestion on some roads, leading to a disruption in the overall flow of traffic. In other words, your job is to alter the $s$--$t$ route gradually (one road at a time), whilst ensuring that road commuters do not have to undertake a longer route from $s$ to $t$ during the process.

A more-or-less similar scenario arises in the case of road accidents~\cite{WangDZM16}. This can sometimes lead to a certain road becoming inoperable, leading to bottleneck situations that could increase the travel times of the commuters. In this case, it should be possible to quickly find a way to reroute the traffic gradually and efficiently.

In $\spr$, only one vertex can be changed at each reconfiguration step, by definition. This condition can be sometimes too restrictive for practical purposes. When a graph is used to model a road network, roads are generally represented by simple induced paths, and vertices on the path represent various landmarks like bus stops, gas stations, shops, etc.~\cite{BastFM06, BastFMSS07, BauerD09, GoldbergKW06}. 

To model the fact that all these consecutive vertices can be changed in one go, we introduce the $\kspr$ problem, where one can change at most $k$ (for some fixed positive integer $k$) contiguous vertices at each reconfiguration step. We study the optimization variant of $\kspr$, where each road has a ``cost of construction'' associated with it and the aim is to produce a reconfiguration sequence whose total construction cost is close to optimal.

We also study $\kspr$ for line graphs and graph powers. These graph classes give us interesting theoretical results that enhance our understanding of $\spr$. Optimization variants of other types of reconfiguration problems (e.g., reconfiguring swarm robots) have also been studied previously~\cite{KirkpatrickL16, DemaineFKMS19}.

\section{Hardness Results}

\begin{definition}[$\kspr$]
Let $(v_1,v_2,\ldots,v_r)$ be an $s$--$t$ shortest path. Then for each $1\leq i<j\leq r$ such that $j-i<k$, one may replace the subpath $(v_i,v_{i+1},\ldots,v_j)$ by a completely new subpath $(u_i,u_{i+1},\ldots,u_j)$ 
in a single reconfiguration step of $\kspr$.
\end{definition}

$\kspr$ for $k=1$ is precisely the $\spr$ problem, which is known to be $\PSPACE$-complete~\cite{Bonsma13}. Note that the $\PSPACE$-hardness of $\spr$ does not straightaway imply the $\PSPACE$-hardness of $\kspr$. We show that $\kspr$ is $\PSPACE$-complete, even for a restricted graph class called line graphs.

\subsection{Hardness of~\texorpdfstring{$\kspr$}{k-SPR} for Line Graphs}
\label{sec:k-spr-line}

In this section, we will see that $\kspr$ (for $k\geq 2$) can be significantly harder than $\spr$. In particular, we show that $\kspr$ (for $k\geq 2$) is $\PSPACE$-complete for line graphs. On the other hand,~\cite{Bonsma13} showed that $\spr$ (or $\kspr$ for $k=1$) can be solved in polynomial time for line graphs (in fact, for claw-free graphs, a superclass of line graphs).

\begin{definition}
Given a graph $G$ on $m$ edges, its line graph $L(G)$ is an $m$-vertex graph where each vertex of $L(G)$ corresponds to an edge of $G$, such that two vertices of $L(G)$ are adjacent if and only if their corresponding edges in $G$ share a vertex (see~\autoref{fig:linegraph1} for an example).
\end{definition}

\begin{figure}[h]
    \centering
    \vspace{0.5cm}
    \includegraphics[page=3, width=0.9\linewidth]{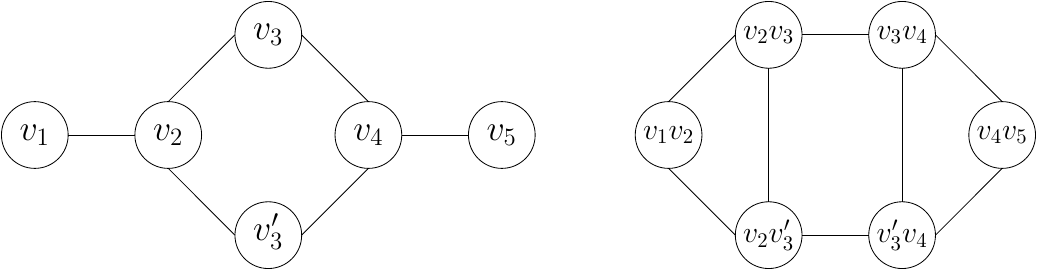}
    \vspace{0.5cm}
    \caption{A graph (left) and its line graph (right)}
    \label{fig:linegraph1}
\end{figure}

\begin{theorem}
\label{lem:ksprlinegraphs}
$\kspr$ is $\PSPACE$-complete for all $k\geq 2$, even when the input graphs are restricted to line graphs.
\end{theorem}

\begin{proof}
Fix an integer $k\geq 2$. We reduce $\spr$ on general graphs to $\kspr$ on line graphs. Consider an $\spr$ instance $(G,s,t,P,Q)$, where $P$ and $Q$ are $s$--$t$ shortest paths in $G$. The goal is to check if $P$ and $Q$ are reconfigurable in $G$. From $(G,s,t,P,Q)$, we construct a $\kspr$ instance $(G',s',t',P',Q')$, where $P'$ and $Q'$ are $s'$-$t'$ shortest paths in $G'$, such that $P'$ and $Q'$ are $k$-reconfigurable in $G'$ if and only if $P$ and $Q$ are reconfigurable in $G$. Also, $G'$ is a line graph that can be constructed from $G$ in polynomial time in three steps, explained below (see~\autoref{fig:line2} for an illustration of these steps).

\begin{figure}
    \centering
    \vspace{0.5cm}
    \includegraphics[width=0.9\linewidth]{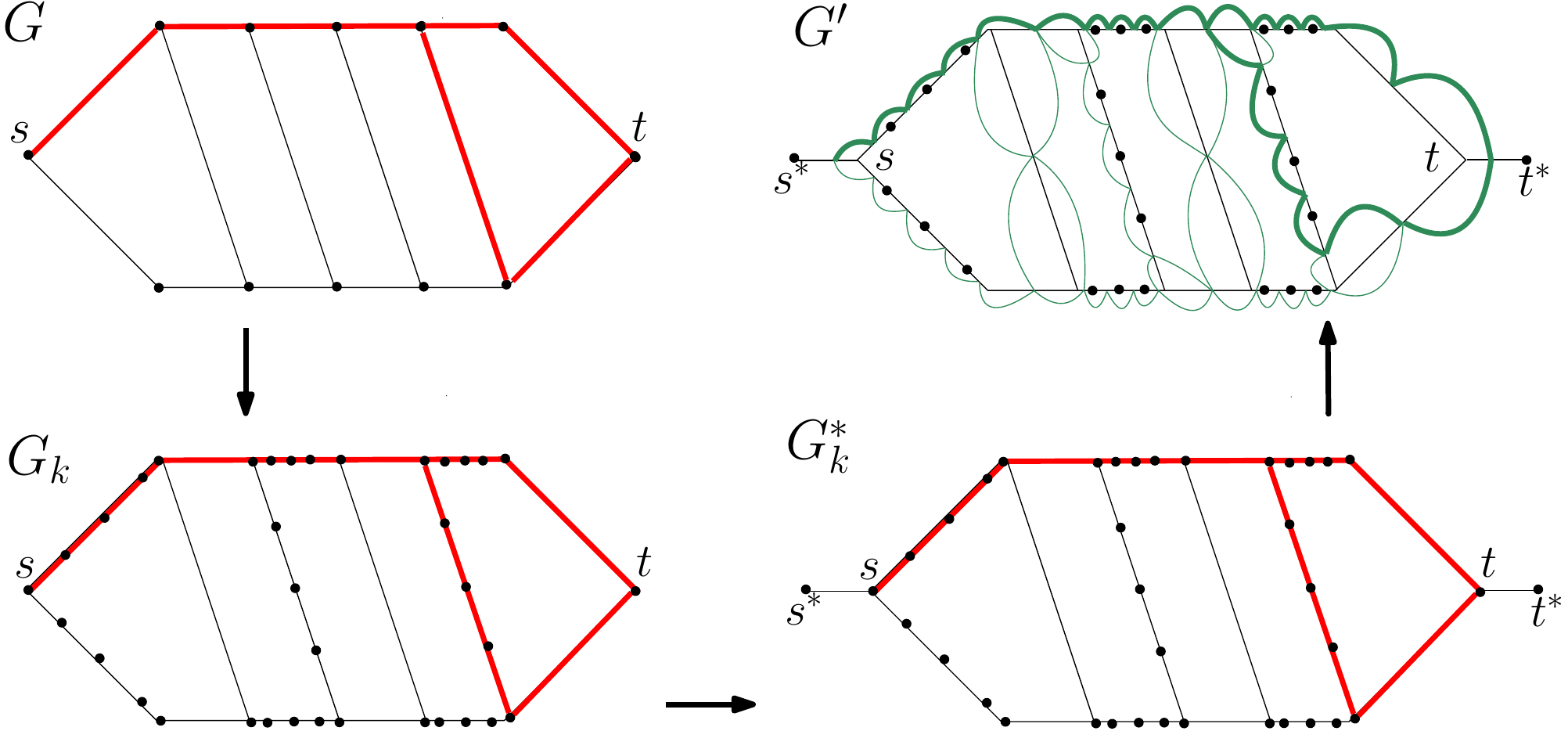}
    \vspace{0.5cm}
    \caption{The construction of $G'=L(G^*_k)$ to show $\PSPACE$-hardness of $\kspr$ for $k=5$. The two $s$--$t$ shortest paths indicated in $G$ (in bold red) differ in only one vertex. Their corresponding $s^*$--$t^*$ shortest paths (in bold green) differ in 5 vertices in $G'$.}
    \label{fig:line2}
\end{figure}

\underline{Step (i):} Consider the layered graph representation of $G$, with $s$ being the zeroth layer and $t$ being the last layer. This can be done by constructing a BFS tree rooted at $s$. 
Now replace every ``even-odd'' edge (i.e., every edge connecting a vertex in layer $i$ to a vertex in layer $i+1$, for every even $i$) by a path on $k$ vertices between the two end points of the edge. Note that if $k=2$, then this last operation does nothing. Let this new graph be $G_k$, and the new paths corresponding to $P$ and $Q$ in $G_k$ be $P_k$ and $Q_k$, respectively.

\underline{Step (ii):} Add two vertices $s^*$ and $t^*$ to $G_k$ such that $s^*$ is adjacent only to $s$, and $t^*$ is adjacent only to $t$. Let this new graph be $G^*_k$. The start vertex of $G^*_k$ is $s^*$ and the target vertex of $G^*_k$ is $t^*$. Thus, each $s$--$t$ shortest paths of $G$ corresponds to an $s^*$--$t^*$ shortest paths of $G^*_k$ whose first edge is always $(s^*,s)$ and last edge is always $(t,t^*)$.

\underline{Step (iii):} Let $G'=L(G^*_k)$. Since $G'$ is the line graph of $G^*_k$, each vertex of $G'$ is labelled by two vertices of $G^*_k$. That is, a vertex $xy$ in $G'$ (where $x$ and $y$ are two adjacent vertices of $G^*_k$) corresponds to an edge $(x,y)$ in $G^*$. The vertex $s^*s$ is our start vertex $s'$ and the vertex $tt^*$ is our target vertex $t'$.

This completes our construction of $G'$. The paths $P'$ and $Q'$ in $G'$ have their first vertex as $s^*s$, and their last vertex as $tt^*$. Their remaining vertices are the edges on the paths $P_k$ and $Q_k$, respectively. Given the fact that $P$ and $Q$ are $s$--$t$ shortest paths in $G$, it is easy to check that $P'$ and $Q'$ are $s'$--$t'$ shortest paths in $G'$. This completes the definition of the $\kspr$ instance $(G',s',t',P',Q')$. 
We make the following claim, whose proof will complete our proof of~\autoref{lem:ksprlinegraphs}.

\begin{claim}
$(G,s,t,P,Q)$ is a yes-instance of $\spr \Longleftrightarrow (G',s',t',P',Q')$ is a yes-instance of $\kspr$.
\end{claim}
\underline{$\Rightarrow$ direction:} Every reconfiguration step in $G$ changes some vertex $u_i$ in layer $i$ to a vertex $v_i$ in the same layer, where $(u_{i-1},u_i,u_{i+1},v_i)$ is a 4-cycle in $G$. Note that $u_i$ can never be $s$ or $t$, so it cannot be present in the zeroth or last layer of $G$. Thus, the graph $G^*_k$ contains a vertex $u_{i-2}$ (possibly $s^*$) and a vertex $u_{i+2}$ (possibly $t^*$). Both $u_{i-2}u_{i-1}$ and $u_{i+1}u_{i+2}$ are vertices in the line graph $G'=L(G^*_k)$. Among the two edges $(u_{i-1},u_i)$ and $(u_i,u_{i+1})$ in $G$, one is retained as an edge in $G_k$ and one is converted to a path on $k$ vertices in $G_k$ (depending on whether $i$ is odd or even). The retained edge contributes to a single vertex in the line graph $G'$, and the path on $k$ vertices (or $k-1$ edges) contributes $k-1$ vertices to the line graph $G'$. Thus, there are $1+(k-1)=k$ vertices between $u_{i-2}u_{i-1}$ and $u_{i+1}u_{i+1}$ on the path $P'$ in $G'$. These $k$ vertices are reconfigured to another set of $k$ vertices on the path $Q'$ in $G'$.

\underline{$\Leftarrow$ direction:} 
Consider a reconfiguration step in $G'$ which replaces a subpath of $j$ (where $j\leq k$) vertices on a shortest $s'$--$t'$ path by another subpath of $j$ vertices. Since $G'$ is the line graph of $G^*_k$, these $j$ vertices of $G'$ can be mapped back to a subpath of $j$ edges in $G^*_k$ (i.e., a subpath of $j+1$ vertices in $G^*_k$). Let $x$ and $y$ be the first and last vertices of the subpath comprised by these $j+1$ vertices in $G^*_k$. It is easy to see that neither $x$ nor $y$ are not changed by mapping the reconfiguration step in $G'$ back to a reconfiguration step in $G^*_k$. 
Note that $x$ is adjacent to at least two vertices in the next layer in $G^*_k$ (thus $x\neq s^*$ and so $x\in G_k$) and $y$ is adjacent to at least two vertices in the previous layer in $G^*_k$ (thus $y\neq t^*$ and so $y\in G_k$). Therefore, $x$ and $y$ can be mapped back to vertices $\mathsf{im}(x)$ and $\mathsf{im}(y)$ in $G$, because all ``new'' vertices of $G_k$ are adjacent to only one vertex in the next layer and only one vertex in the previous layer. Finally, if $\mathsf{im}(x)$ is in layer $i$ of $G$ (for some $i$), then $\mathsf{im}(y)$ must be in layer $i+2$ of $G$. This is because if $\mathsf{im}(y)$ lies in a layer before $i+2$ (i.e., $i+1$), then $(\mathsf{im}(x),\mathsf{im}(y))$ would be a multiple edge in $G$, which is a contradiction. And if $\mathsf{im}(y)$ lies in a layer after $i+2$, then $x$ and $y$ would have more than $k$ edges between them in $G_k$. This is also a contradiction, since $j\leq k$.
\end{proof}

\subsection{Hardness of~\texorpdfstring{$\spr$}{SPR} for Graph Powers}
\label{sec:k-spr-power}

In this section, we will show that it is possible to use the $\PSPACE$-hardness of $\kspr$ to prove $\PSPACE$-hardness of $\spr$ for some graph classes, namely graph powers.

\begin{definition}
The $k$\textsuperscript{th} power of a graph $G$ is obtained by making all vertices $u,v$ such that $d(u,v)\leq k$ adjacent.
\end{definition}

\begin{theorem}
$\spr$ is $\PSPACE$-complete for graph powers.
\end{theorem}

\begin{proof}
Our proof technique is as follows. Let $G^k$ be the $k$\textsuperscript{th} graph power of $G$. We use the $\PSPACE$-hardness of $(2k-1)$-SPR for $G$ to show the $\PSPACE$-hardness of SPR (or 1-SPR) for $G^k$.

Let $P_1$, $P_2$  be two reconfigurable $s$--$t$ shortest paths in $G^k$. Consider the reconfiguration sequence. At each step of the reconfiguration sequence we replace a vertex $v$ with $v'$ where both of them have edges to $u$ and $w$ that belong to intermediate $s$--$t$ shortest path. 
We can construct a reconfiguration sequence for two $s$--$t$ shortest paths $P'_1, P'_2$ in $G$, where $V(P_1) \subseteq V(P'_1)$ and $V(P_2) \subseteq V(P'_2) $, as follows:
\begin{itemize}
    \item If the edges $uv, vw, uv', v'w \in E(G)$ then reconfiguration step remain unchanged
    \item If any of the $uv, vw, uv', v'w \in E(G^k)$ then consider a shortest path between the two vertices in $G$. Replace all of them in a single step.
\end{itemize}
Clearly, following the above steps it is possible to $P'_2$ from $P'_1$. The number of vertices we change in the second case is at most $2k-1$ as an edge in $G^k$ implies that there exists a path of length at most $k$ between them in $G$. If both edges involved in the reconfiguration step from $v$ or $v'$ are in $G$, then we change at most $2k-1$ vertices in one step. Hence, it is a $(2k-1)$-$\spr$.  

To prove the other direction, let $P_1'$ and $P_2'$ be two reconfigurable $s$--$t$ shortest paths in $G$. Each reconfiguration step in $G$ changes at most $(2k-1)$ contiguous vertices. Fix a reconfiguration step that changes $k'$ contiguous vertices between two vertices $u$ and $v$ in $G$. Clearly, $k'\leq 2k-1$. Consider the $\lfloor k'/2\rfloor$\textsuperscript{th} vertex after $u$ on $P_1$, and call it $x$. Similarly consider the $\lfloor k'/2\rfloor$\textsuperscript{th} vertex after $u$ on $P_2$, and call it $y$. Note that $(u,x)$, $(x,v)$ are edges in $G^k$, as are $(u,y)$ and $(y,v)$. Thus, they can trivially be reconfigured in $G^k$. This completes the proof. 
\end{proof}

\section{Polynomial-time Solvable~\texorpdfstring{$\spr$}{SPR} Problems}
In this section, we present polynomial-time algorithms for $\spr$ on circle graphs, bridged graphs, Boolean hypercubes and graphs of constant diameter.

\subsection{\texorpdfstring{$\spr$}{SPR} for Permutation Graphs}
\label{sec:permute}

Permutation graphs are a subclass of circle graphs, and therefore the algorithm for circle graphs that we present in the next section (\autoref{sec:circle}) also holds for permutation graphs. However, it is instructive to study the polynomial-time algorithm for permutation graphs first, and then generalise it to circle graphs, as the circle graphs algorithm borrows several ideas from the permutation graphs algorithm.

\begin{definition}
A graph $G$ on $n$ vertices $\{1,2,\ldots,n\}$ is called ermutation graph if there exists a permutation $\sigma = (\sigma(1)\sigma(2)\cdots\sigma(n)) \in \mathcal{S}_n$ such that for every $1\leq i<j\leq n$, we have $(i,j)\in E(G)$ if and only if $\sigma(i)>\sigma(j)$.
\end{definition}

Given a graph, its permutation representation can be constructed in linear time, if it exists~\cite{Golumbic}. Let $G$ be a permutation graph on $n$ vertices with two special vertices $s$ and $t$. We delete all the edges of $G$ which do not lie on an $s$--$t$ shortest path, and label the edges of $G$ as L-type (L stands for left) or R-type (R stands for right) as follows.

\begin{definition}
Let $e$ be an edge of an $s$--$t$ shortest path $P=(s=v_{1}, v_{2}, \ldots, v_{l}=t)$ in $G$. Then $e=(v_i,v_{i+1})$ for some $1\leq i\leq l-1$. We say that $e$ is of L-type if $v_{i+1}<v_i$, and of R-type if $v_i<v_{i+1}$. (See~\autoref{fig:permutationLR} for an example.)
\end{definition}

\begin{figure}[t]
    \centering
    \vspace{0.5cm}
    \includegraphics[scale = 0.8]{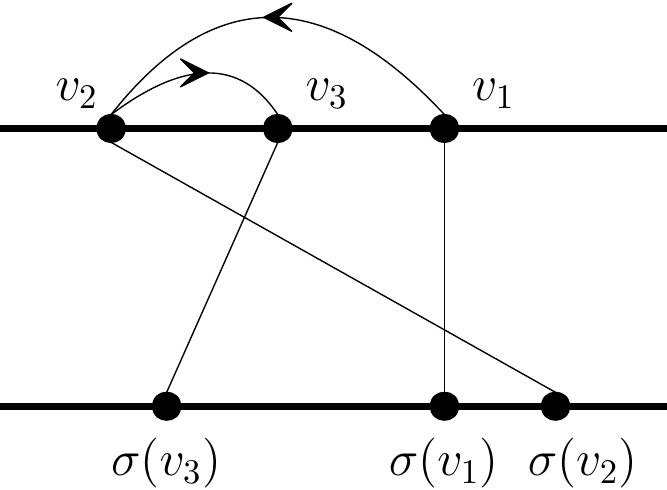}
    \vspace{0.5cm}
    \caption{$(v_1,v_2)$ is of L-type and $(v_2,v_3)$ is of R-type}
    \label{fig:permutationLR}
\end{figure}

Throughout this proof, we will assume that $d(s,t)\geq 3$ (that is, $l\geq 4$) and that $t$ is to the left of $s$, as shown in~\autoref{fig:permutationst}. This means that for every $3\leq i\leq l$,
\begin{equation} 
\label{eq:ttoleftofs}
    (v_i,s)\notin E(G),\ v_i<s \text{ and } \sigma(v_i)<\sigma(s).
\end{equation}

\begin{lemma} 
\label{lem:leftright}
Let $P=(s=v_{1}, v_{2}, \ldots, v_{l}=t)$ be a shortest path in $G$. Then for every $2\leq i\leq l-1$, the edges $(v_{i-1},v_i)$ and $(v_i,v_{i+1})$ are of different types.
\end{lemma}

\begin{proof} We will prove this lemma by contradiction. Suppose $(v_{i-1},v_i)$ and $(v_i,v_{i+1})$ are of the same type (say L-type). Then it is easy to see the following.
\begin{align*}
    &v_{i+1}<v_i<v_{i-1};\\
    &\sigma(v_{i-1})<\sigma(v_i)<\sigma(v_{i+1}).
\end{align*}
These imply that $(v_{i-1},v_{i+1})$ is also an edge, which is impossible in a shortest path.
\end{proof}

\begin{lemma} \label{lm:perm}
Two shortest paths $P_1$ and $P_2$ can be reconfigured in $G$ if and only if the first edge on both the paths is of the same type. Further, the reconfiguration sequence can be obtained in linear time.
\end{lemma}

\begin{proof}
Let the paths be
\begin{align*}
    P_1&=(v_1,v_2,v_3,\ldots,v_{l-1},v_l)\ \text{ and}\\
    P_2&=(u_1,u_2,u_3,\ldots,u_{l-1},u_l),
\end{align*}
where $v_1=u_1=s$ and $v_l=u_l=t$.

\underline{$\Leftarrow$ direction:} We will show that if $(s,v_2)$ and $(s,u_2)$ are of the same type (say L-type), then $P_1$ and $P_2$ can be reconfigured. Our proof is by induction on $\ell$ (note that $\ell=d(s,t)-1$). For $\ell=3$, this is trivial (the paths are $(s,v_2,t)$ and $(s,u_2,t)$). Now assume that every pair of shortest paths in permutation graphs with $l-1$ vertices each, both starting with an L-type edge, is reconfigurable.

We will show that $P_1$ and $P_2$, two paths  with $l$ vertices each, both starting with an L-type edge, are always reconfigurable. If $v_2=u_2$, then let $s'=v_2=u_2$, and the $s'$-$t$ subpaths of $P_1$ and $P_2$ have $l-1$ vertices each, which can be reconfigured by the induction hypothesis, and we are done. Now, if $v_2\neq u_2$, assume that $v_2<u_2$ (the proof for $u_2<v_2$ is similar). It is helpful to follow~\autoref{fig:permutationst} while reading the rest of this proof. Since both $(s,v_2)$ and $(s,u_2)$ are L-type edges, this means that both $(v_2,v_3)$ and $(u_2,u_3)$ are R-type edges, by~\autoref{lem:leftright}. We have the following.
\begin{align*}
    &v_2<u_2<s; &&\text{($(s,v_2)$ and $(s,u_2)$ are L-type)}\\
    &u_2<u_3; &&\text{($(u_2,u_3)$ is R-type)}\\
    &u_3<s,\ \sigma(u_3)<\sigma(s); &&\text{(substitute $i=3$ in~\eqref{eq:ttoleftofs})}\\
    &\sigma(s)<\sigma(v_2). &&\text{($(s,v_2)$ is an edge)}
\end{align*}
The first two lines imply that $v_2<u_3$ and the last two lines imply that $\sigma(u_3)<\sigma(v_2)$. Thus $(v_2,u_3)$ is an edge in $G$. We can therefore reconfigure the path $P_2$ by replacing $u_2$ with $v_2$, obtaining $P_2'=(s,v_2,u_3,\ldots,t)$. Now, setting $s'=v_2$ in both $P_1$ and $P_2'$, we obtain two $s'$-$t$ paths $(s',v_3,\ldots,t)$ and $(s',u_3,\ldots,t)$, each with $l-1$ vertices. By the induction hypothesis, these can be reconfigured. Note that our proof also implies that the reconfiguration sequence can be obtained in linear time.

\begin{figure}[b]
    \centering
    \vspace{0.5cm}
    \includegraphics[width = 0.8\linewidth]{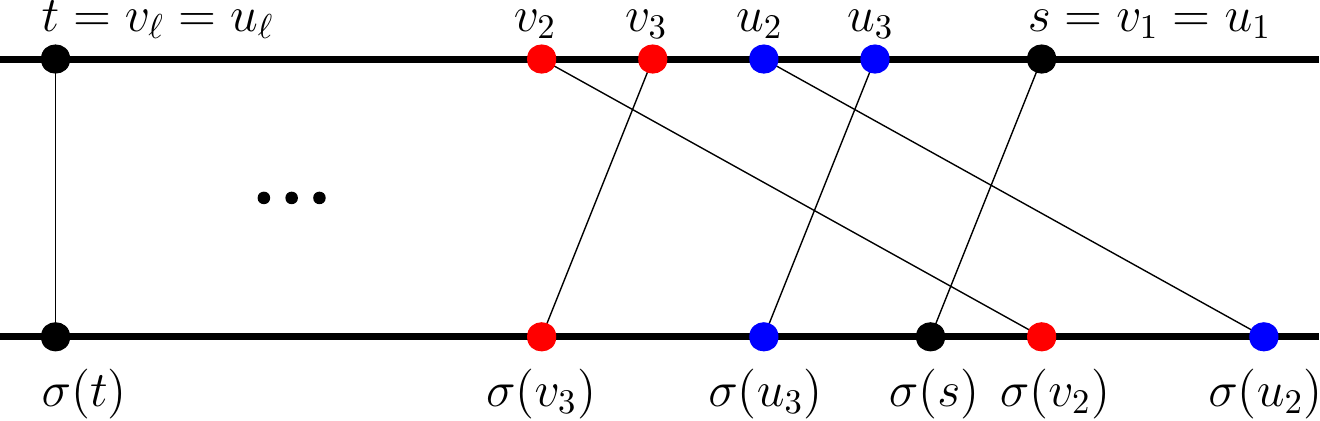}
    \vspace{0.5cm}
    \caption{$s$ and $t$ are in black, vertices of $P_1$ are in red and vertices of $P_2$ are in blue}
    \label{fig:permutationst}
\end{figure}

\underline{$\Rightarrow$ direction:} We will show that if $P_1$ can be reconfigured to $P_2$, then their first edges are of the same type. For this, we will simply show that a reconfiguration step does not change the type of the first edge.

Consider a reconfiguration step in which a vertex $v_i$ of $P_1$ is changed to a vertex $w_i$ ($w_i$ may or may not be $u_i$). If $i>2$, then this clearly does not change the type of the first edge. If $i=2$, then the new path is
\begin{equation*}
    P_1'=(s,w_2,v_3,v_4,\ldots,v_{l-1}, v_l).
\end{equation*}
Let the first edge $(s,v_2)$ of $P_1$ be of L-type (the proof for R-type is similar). We will show that the first edge $(s,w_2)$ of $P_1'$ is also of L-type. By~\autoref{lem:leftright}, the edges $(v_2,v_3)$ and $(v_3,v_4)$ in $P_1$ are of R-type and L-type, respectively. In $P_1'$, since $(v_3,v_4)$ is of L-type, we have $(w_2,v_3)$ is of R-type and $(s,w_2)$ is of L-type (again by~\autoref{lem:leftright}).
\end{proof}

\subsection{\texorpdfstring{$\spr$}{SPR} for Circle Graphs}
\label{sec:circle}
\begin{figure}[ht]
    \centering
    \vspace{0.5cm}
    \includegraphics[width = 0.9\linewidth]{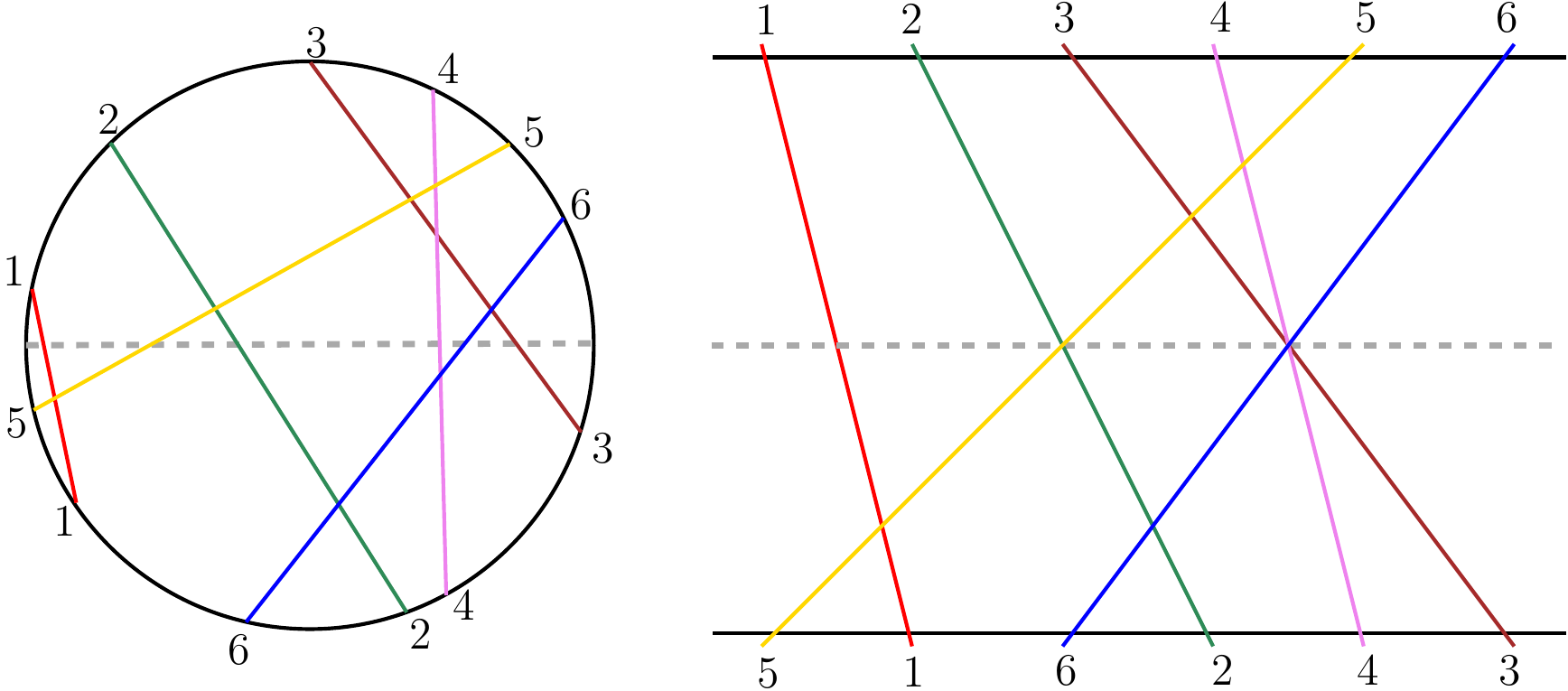}
    \vspace{0.5cm}
    \caption{A circle graph with an equator (indicated by the dotted line) and its corresponding permutation graph}
    \label{fig:circleperm}
\end{figure}
We will now show that $\spr$ can be solved in linear time for a superclass of permutation graphs called~\emph{circle graphs}, using much of the same ideas we did for permutation graphs.

\begin{definition} A graph is called a circle graph if its vertices can be represented by the chords of a circle such that two vertices have an edge in the graph if and only if their corresponding chords intersect.
\end{definition}

The following well-known fact establishes the connection between circle graphs and permutations graphs.

\begin{fact}[\cite{Brandstadt}] \label{fact:circle} A graph is a permutation graph if and only if it is a circle graph that admits an equator, i.e., an additional chord that intersects every other chord of the circle graph (\autoref{fig:circleperm}).
\end{fact}

Given a graph, its circle representation can be constructed in quadratic time, if it exists~\cite{Golumbic}. 
Let $G$ be a circle graph and $P_1$ and $P_2$ be two $s$--$t$ shortest paths in $G$. For every vertex $v$ we assign it a label $i$, if it appears on the $i$\textsuperscript{th} level of the $\bfs$ tree rooted at $s$. A chord with label $i$ intersects one or more chords with label $i-1$ and one or more chords with label $i+1$ (possibly even other chords at label $i$, but we ignore those for our proof). We orient the chord $i$~\emph{from} the point of intersection of the $i-1$ chord on chord $i$ (called the first end point)~\emph{to} the point of intersection of the $i+1$ chord on the chord $i$ (called the second end point).

\begin{figure}[ht]
    \centering
    \vspace{0.5cm}
    \includegraphics[width=0.9\linewidth]{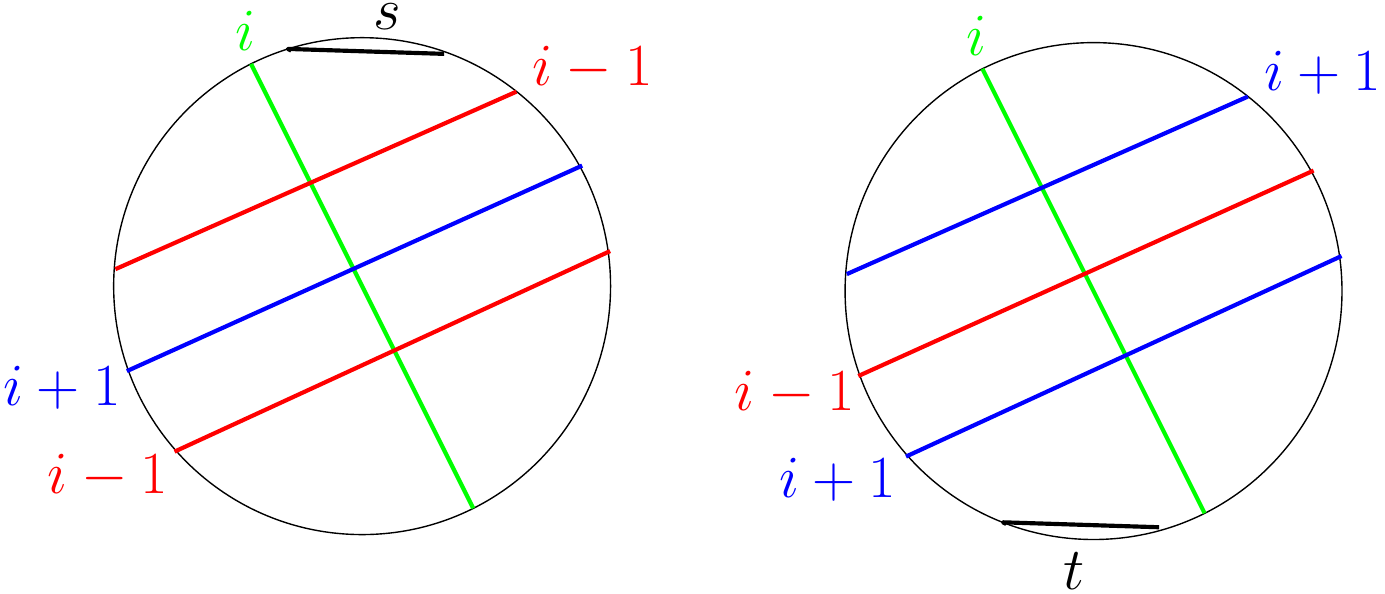}
    \vspace{0.5cm}
    \caption{Proof of~\autoref{cl:chords}}
    \label{fig:circle2}
\end{figure}

\begin{lemma} \label{cl:chords}
Every chord has a unique orientation. In other words, a chord cannot receive two different orientations from two different shortest paths.
\end{lemma}
\begin{proof}
For the sake of contradiction, assume that there exists a chord $v$ (say, labelled $i$) with two orientations. This means that there exists a pair of chords with labels $(i-1,i+1)$ which orient $v$ in one direction, and also another pair of chords with labels $(i-1,i+1)$ which orient $v$ in the opposite direction. A little bit of thought reveals that the only way this can happen is if at least one of the following is true.
\begin{enumerate}
    \item The intersection point of a chord of label $i+1$ lies between the intersection points of two chords of label $i-1$.
    \item The intersection point of a chord of label $i-1$ lies between the intersection points of two chords of label $i+1$.
\end{enumerate}
In the first case (\autoref{fig:circle2}, left), let $u_{i-1}$ and $u'_{i-1}$ be the two chords at level $i-1$, and $u_{i+1}$ be the chord at level $i+1$. The chord $u_{i+1}$ divides the circle into two parts. It is easy to see that $u_{i-1}$ and $u'_{i-1}$ lie in opposite parts. Since $s$ does not intersect $u_{i+1}$, it also lies in one of these two parts (say, in the same part as $u_{i-1}$). Now, consider the chords on a shortest path from $s$ to $u'_{i-1}$. Note that at least one of these chords must intersect $u_{i+1}$, implying that $d(s,u_{i+1})\leq d(s,u'_{i-1})$. This is clearly a contradiction.

In the first case (\autoref{fig:circle2}, right), let $u_{i+1}$ and $u'_{i+1}$ be the two chords at level $i+1$, and $u_{i-1}$ be the chord at level $i-1$. The chord $u_{i-1}$ divides the circle into two parts. It is easy to see that $u_{i+1}$ and $u'_{i+1}$ lie in opposite parts. Since $t$ does not intersect $u_{i-1}$, it also lies in one of these two parts (say, in the same part as $u_{i+1}$). Now, consider the chords on a shortest path from $u'_{i+1}$ to $t$. Note that at least one of these chords must intersect $u_{i-1}$, implying that $d(u_{i-1},t)\leq d(u'_{i+1},t)$. This is clearly a contradiction.
\end{proof}

If the circle graph admits an equator, then we can directly use the algorithm for permutation graphs (\autoref{fact:circle}) from the previous section (\autoref{sec:permute}). Otherwise, we define the equator to be an additional chord intersecting both $s$ and $t$. Given~\autoref{cl:chords}, we will now define a labelling scheme of the chords based on their orientation and also on the positions of their end points with respect to the equator. Each chord receives one of four possible labels. (This is similar to the L-type and R-type labels in permutation graphs.)
\begin{itemize}
    \item $\ttt$ if both end points of the chord are above the equator.
    \item $\tb$ if the first end point of the chord is above the equator and the second end point is below the equator.
    \item$\bt$ if the first end point of the chord is below the equator and the second end point is above the equator.
    \item $\bb$ if both end points of the chord are below the equator.
\end{itemize}

\begin{figure}[ht]
    \centering
    \vspace{0.5cm}
    \includegraphics[width=0.6\linewidth]{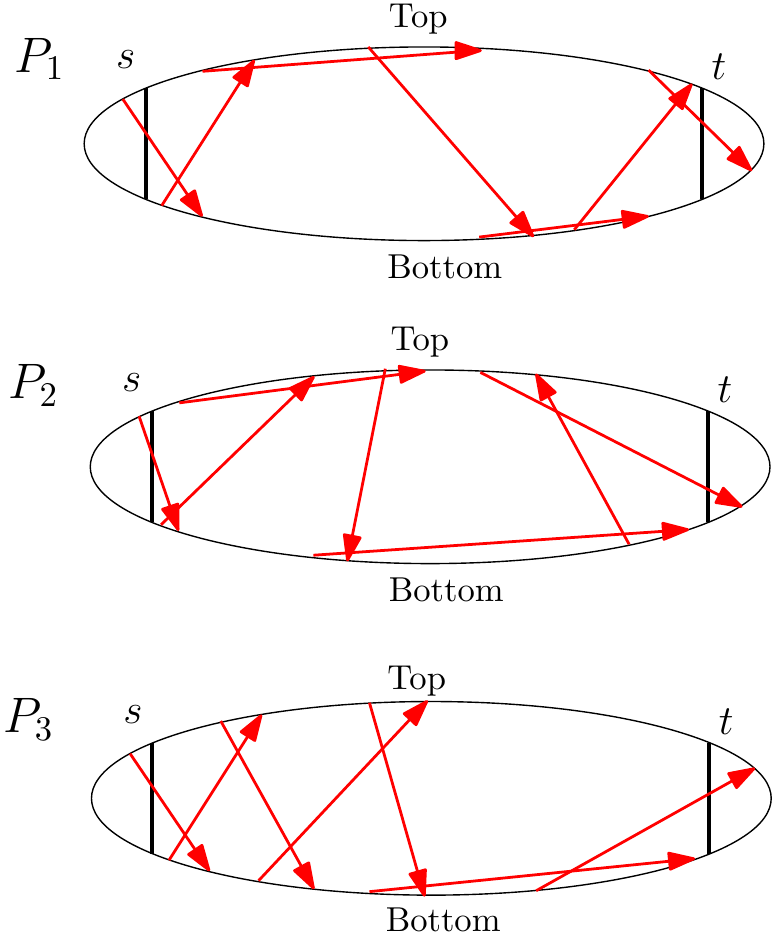}
    \vspace{0.5cm}
    \caption{Three $s$--$t$ shortest paths in a circle graph}
    \label{fig:ellipse}
\end{figure}

Here, T stands for top and B stands for bottom. The label of an $s$--$t$ shortest path is simply a concatenation of the labels of its vertices from $s$ to $t$. In~\autoref{fig:ellipse}, the labels are
\begin{align*}
    P_1 &= (\ssb) (\bt) (\ttt) (\tb) (\bb) (\bt) (\tst)\\
    P_2 &= (\ssb) (\bt) (\ttt) (\tb) (\bb) (\bt) (\tst)\\
    P_3 &= (\ssb) (\bt) (\tb) (\bt) (\tb) (\bb) (\bst)
\end{align*}
We are now set to prove our main theorem.
\begin{theorem} \label{thm:circle}
Two $s$--$t$ shortest paths in a circle graph are reconfigurable if and only they have the same label. Furthermore, the reconfiguration sequence can be obtained in linear time (if it exists). 
\end{theorem}

\begin{proof}
\underline{$\Rightarrow$ direction:} We will see that if two $s$--$t$ shortest paths can be reconfigured, they have the same label. Note that the label of each vertex matches up with the label of the vertex before it and the label of the vertex after it. For example, in $P_3$, the fourth and sixth vertices are $(\bt)$ and $(\bb)$. The first end point of the fifth vertex must match the second end point of the fourth vertex, and the second end point of the fifth vertex must match the first end point of the sixth vertex. So the label of the fifth vertex must be $(\tb)$. Thus, a reconfiguration step that changes the fifth vertex of $P_3$ can only change it to another vertex whose label is also $(\tb)$. Hence, reconfiguration does not change the label of a path. (In the example, $P_3$ cannot be reconfigured to $P_1$ or $P_2$.)

\underline{$\Leftarrow$ direction:} We will see that if two $s$--$t$ shortest paths have the same label, they can be reconfigured. In the example shown, let $v_1$ and $u_1$ be the first vertices of $P_1$ and $P_2$, respectively. Both are labelled $(\ssb)$. Suppose the second end point of the chord $u_1$ lies before (to the left of) the second end point of the chord $v_1$ on the bottom of the circle. Then, the chord $u_2$ must intersect the chord $v_1$. This means the $(s,u_1,u_2)$ subpath can be changed to $(s,v_1,u_2)$ by reconfiguring $u_1$ to $v_1$. Next, we can show that either the chord $u_3$ must intersect the chord $v_2$ or the chord $v_3$ must intersect the chord $u_2$, and similar to the proof of~\autoref{lm:perm}, we will eventually reconfigure $P_1$ and $P_2$. Also, it is easy to see that the reconfiguration sequence thus obtained is of size $O(n)$.
\end{proof}

\subsection{\texorpdfstring{$\spr$}{SPR} for Bridged Graphs}
\label{sec:bridged}
We begin this section with some definitions. Let $G$ be a graph, and $u,v$ be two vertices of $G$. Their~\emph{interval} $I(u,v)$ is the set of all  vertices of $G$ that lie on at least one shortest $u$-$v$ path. More formally, $$I(u,v) = \{w \in V(G) \colon d(u,w) + d(w,v) = d(u,v)\}.$$

A subset of vertices $H$ of $V(G)$ is called~\emph{convex} if for each pair of vertices $(u, v) \in H\times H$, their interval $I(u,v) \subseteq H$.

\begin{definition}
A graph $G$ is called a~\emph{bridged} graph if the neighbourhood of every convex set in $G$ is also convex.
\end{definition}

It is known that bridged graphs are precisely the graphs in which all isometric cycles have length three~\cite{SoltanC83, FarberJ87}. In particular, all chordal graphs are bridged. Bonsma~\cite{Bonsma13} showed that $\spr$ can be solved in polynomial time for chordal graphs. We extend Bonsma's result to bridged graphs.

Let us now look at some properties of bridged graphs. A graph is called~\emph{weakly modular} if it satisfies the following two conditions~\cite{BandeltC96, Chepoi89}.

\begin{itemize}
    \item \textbf{The quadrangle condition:} $\forall u, v, w, z \in V(G)$ with $k\coloneqq d(u,v) = d(u,w)$, $d(u,z) = k+1$ and $vz, wz \in E(G)$, $\exists x \in V(G)$ such that $d(u,x) = k-1$ and $xv, xw \in E(G)$.
    \item \textbf{The triangle condition:} $\forall u,v, w \in V(G)$ with $k \coloneqq d(u,v) = d(u,w)$ and $vw \in E(G)$, $\exists x\in V(G)$ such that $d(u,x) = k-1$ and $xv, xw \in E(G)$.
\end{itemize}

Bridged graphs are weakly modular graphs with no induced cycle of length four or five~\cite{Chepoi89}. We essentially present a polynomial-time algorithm for $\spr$ for weakly modular graphs. Our algorithm recursively uses the triangle condition from the above definition. For general graphs, such a recursion would make the running time exponential. We use a suitable data structure to make the running time to polynomial. 

Before describing the algorithm and the proof lets us define the following notation. Let $P_{s,w,t}$ denotes a shortest path between $s$ and $t$ which is going trough the vertex $w$. 

\begin{algorithm}
\caption{$\spr$ for weakly modular graphs}
\label{algo:weakly-modular}
\textbf{Input:} $G$, paths $P_{s,u_\ell,t}$ and $P_{s,v_\ell, t}$,
\begin{algorithmic}[1]
\If{$w_{\ell-1} \in P_{s,u_\ell,t}$}
\State \textbf{Output} $u_l \rightarrow v_l$, $\spr(G, P_{s, w_{\ell -1}, v_\ell}, P_{s, v_{\ell -1}, v_\ell})$ 
\EndIf
\If{$w_{\ell-1} \in P_{s,v_\ell,t}$}
\State \textbf{Output} $\spr(G, P_{s, u_{\ell -1}, u_\ell}, P_{s, w_{\ell -1}, u_\ell})$, $u_l \rightarrow v_l$
\EndIf
\If{$w_{\ell-1} \notin P_{s,u_\ell,t}, P_{s,v_\ell,t}$}
\State \textbf{Output} $\spr(G, P_{s, u_{\ell -1}, u_\ell}, P_{s, w_{\ell -1}, u_\ell})$, $u_l \rightarrow v_l$, $\spr(G, P_{s, w_{\ell -1}, v_\ell}, P_{s, v_{\ell -1}, v_\ell})$
\EndIf
\end{algorithmic}
\end{algorithm}

\begin{lemma}
\label{lem:weakly-modular}
\autoref{algo:weakly-modular} solves $\spr$ on weakly modular graphs in $O(2^\ell n)$ time, where $\ell$ is the distance between $s$ and $t$.
\end{lemma}

\begin{proof}
From the triangle condition of weakly modular graphs, we know that for a given $P_{s,u_\ell,t}, P_{s,v_\ell, t}, s, t$ there exists a vertex $w$ such that both the edges $u_\ell w$ and $v_\ell w$ are preset in the graph. 
Consider a solution for the $\spr$. In that solution we move $u_\ell \rightarrow v_\ell$ at some step. Then, what remains is a solution to $\spr$ on paths whose length is reduced by $1$. Depending on the fact whether $w$ is either $u_{\ell-1}$ or $v_{\ell-1}$ or $w \notin P_{s,u_\ell,t}, P_{s,v_\ell, t}$ we have three subproblems. 
That is precisely what~\autoref{algo:weakly-modular} computes.

At every step searching for a $w$ requires $O(n)$ time. 
Number of the subproblem in the recursion is at most $2^\ell$. Hence total running time is $O(2^\ell n)$. 
\end{proof}

The running time of~\autoref{algo:weakly-modular} is clearly exponential when $\ell=\Theta(n)$. This can be improved. Consider the following data structure.

\begin{definition}
$\look(u_i, v_i)$
\begin{itemize}
    \item Takes input $u_i$, $v_i$ from the same $\bfs$ layer computed from $s$.
    \item Outputs $w_{i-1}$ such that both $w_{i-1}u_i$ and $w_{i-1},v_i \in E(G)$. 
\end{itemize}
\end{definition}

We construct $\look(u_i, v_i)$ by searching for common parents for every pair of vertices in a $\bfs$ layer. Implementing $\look(u_i, v_i)$ takes $O(n^3)$ space. Finding a $w$ at each step using this data structure requires only a constant amount of time. Finally, the $\bfs$ naturally partitions the vertices of $G$ into layers, reducing the running time to $O(n^2)$. (This lookup table method is essentially a dynamic program.) We conclude this section with the following theorem and its corollary.

\begin{figure}
        \centering
        \vspace{0.5cm}
        \includegraphics[width=0.33\linewidth, page=2]{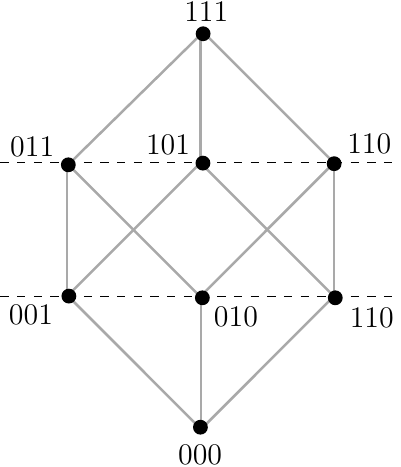}
        \vspace{0.5cm}
        \caption{Boolean hypercube for $d=3$}
        \label{fig:hypercube}
\end{figure}

\begin{theorem} $\spr$ can be solved in $O(n^2)$ time for weakly modular graphs.
\end{theorem}

\begin{corollary} $\spr$ can be solved in $O(n^2)$ time for bridged graphs.
\end{corollary}

\subsection{\texorpdfstring{$\spr$}{SPR} for Boolean Hypercubes}
\label{sec:boolean}

\begin{definition}
A d-dimensional Boolean hypercube is a graph with vertex set $\{0,1\}^d$ such that two vertices are adjacent if and only if their corresponding bit strings differ in exactly one of the $d$ coordinates (\autoref{fig:hypercube}).
\end{definition}

As an input to the $\spr$ problem, we are given two $s$--$t$ shortest paths $P_{1}$ and $P_{2}$ of length $\ell$ each in a $d$-dimensional Boolean hypercube. Let $\oplus$ denote the bit-wise $\xor$ operation and $\ones(\ell)$ denote the positions of the ones in the bit string $\ell$. For example, for $d=5$,
\begin{align*}
s &= 00101\\
t &= 10011\\
s\oplus t &= 10110\\
\ones(s \oplus t) &= \{1,3,4\}.
\end{align*}
Given $\ones(s \oplus t)$ in this example, it is easy to see that every $s$--$t$ shortest path has three edges, one edge each dedicated to changing the bit in the first, third and fourth positions. However, the order in which these changes are made could be different. There are $3!=6$ ways to do this: $(134)$, $(143)$, $(314)$, $(341)$, $(413)$, $(431)$. In other words, there are six $s$--$t$ shortest paths in this example (see~\autoref{fig:hybercube2}). Thus, we will represent all $s$--$t$ shortest paths as permutations for the rest of this proof.

\begin{figure}[h]
        \centering
        \vspace{0.5cm}
        \includegraphics[width=0.55\linewidth]{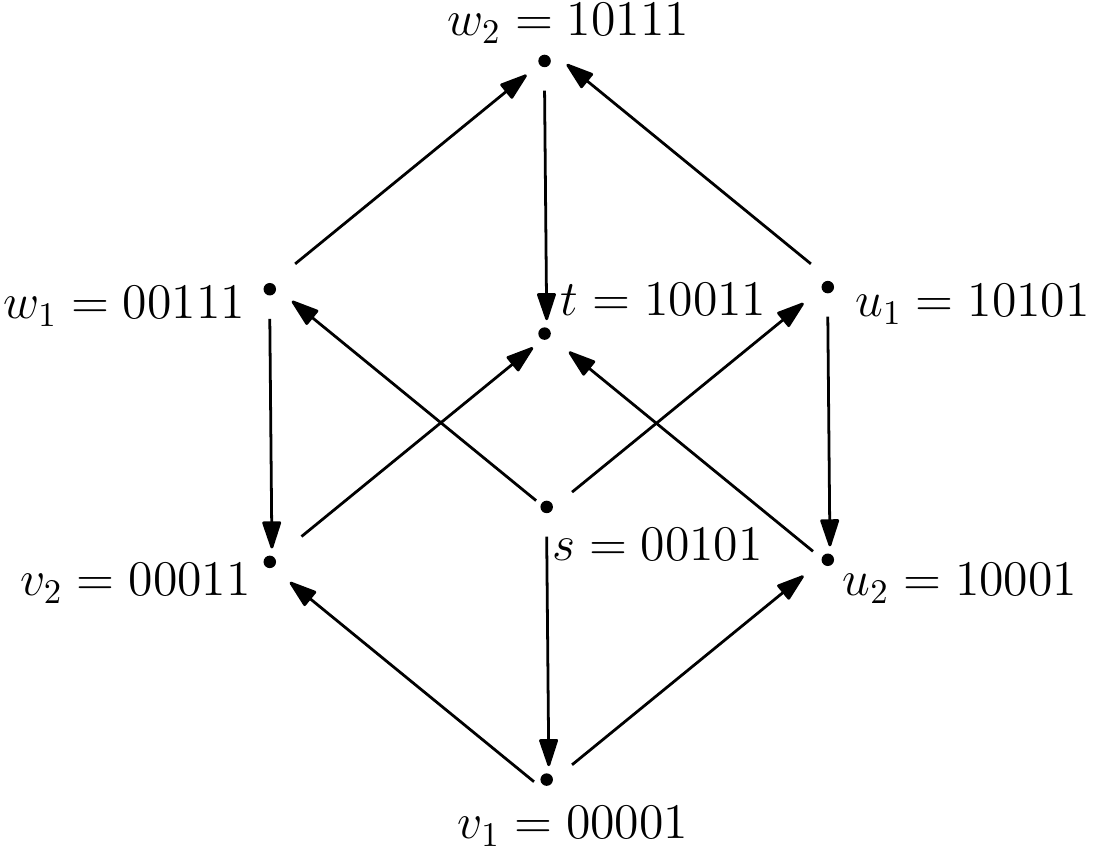}
        \vspace{0.5cm}
        \caption{There are $6$ possible ways to go from $s$ to $t$}
        \label{fig:hybercube2}
\end{figure}

\begin{observation} \label{obs:boohypcub}
Let $s$ and $t$ be two vertices of a Boolean hypercube such that $d(s,t)=r$. Then, two permutations $P_1$ and $P_2$ can be reconfigured in a single reconfiguration step if and only if there exists a $j\in[r-1]$ such that
\begin{align*}
P_{1}=(i_{1}i_{2}i_{3}\ldots i_{j-1}i_{j}i_{j + 1}i_{j+2}\ldots i_{r-1}i_{r});\\
P_{2}=(i_{1}i_{2}i_{3}\ldots i_{j-1}i_{j + 1}i_{j}i_{j+2}\ldots i_{r-1}i_{r}).
\end{align*}
That is, $P_1$ and $P_2$ differ only in the positions $j$ and $j+1$.
\end{observation}
\begin{algorithm}
\caption{$\spr$ for Boolean hypercubes with $r=d(s,t)$}
\label{algo:hypercube}
\textbf{Input:} $\text{Permutations } P_1, P_2\in\mathcal{S}_r$
\begin{algorithmic}[1]
\While{$\exists\ i\in\{1,2,\ldots,r-1\}$ such that $P_2^{-1}(P_1[i]) > P_2^{-1}(P_1[i+1])$}
\State Swap $P_1[i]$ and $P_1[i+1]$ in $P_1$
\EndWhile
\end{algorithmic}
\end{algorithm}

\begin{theorem} \label{thm:hypcub}
\autoref{algo:hypercube} reconfigures two given $s$--$t$ shortest paths (or permutations) $P_{1}$ and $P_{2}$ in a Boolean hypercube in the minimum number of reconfiguration steps.
\end{theorem}

\begin{proof} It is easy to see that~\autoref{algo:hypercube} reconfigures $P_{1}$ to $P_{2}$ by changing~\emph{one vertex} of $P_1$ at each reconfiguration step (\autoref{obs:boohypcub}) until $P_1=P_2$.

For each $\ell_{1},\ell_{2}\in[r]$, if the relative order of $\ell_1$ and $\ell_2$ is the same in both $P_1$ and $P_2$, then note that~\autoref{algo:hypercube} does not change their relative order. Otherwise,~\autoref{algo:hypercube} performs an~\emph{inversion}: it swaps (or exchanges) them, inverting their relative order in $P_1$. As every reconfiguration step corrects at most one such inversion (\autoref{obs:boohypcub}), the number of reconfiguration steps required to reconfigure $P_{1}$ to $P_{2}$ is at least the number of inversions. The total number of inversions between $P_1$ and $P_2$ is in fact called Kendall's Tau distance (or bubble sort distance), a well-known measure of dissimilarity between permutations~\cite{Sedgewick}. This proves the optimality of~\autoref{algo:hypercube}.
\end{proof}

\subsection{\texorpdfstring{$\spr$}{SPR} for Circular-Arc Graphs} \label{sec:circarc}
\begin{definition} A circular-arc graph is the intersection graph of a set of arcs on a circle (\autoref{fig:circulararc}).
\end{definition}

\begin{figure}[ht]
    \centering
    \vspace{0.5cm}
    \includegraphics[width=0.5\linewidth]{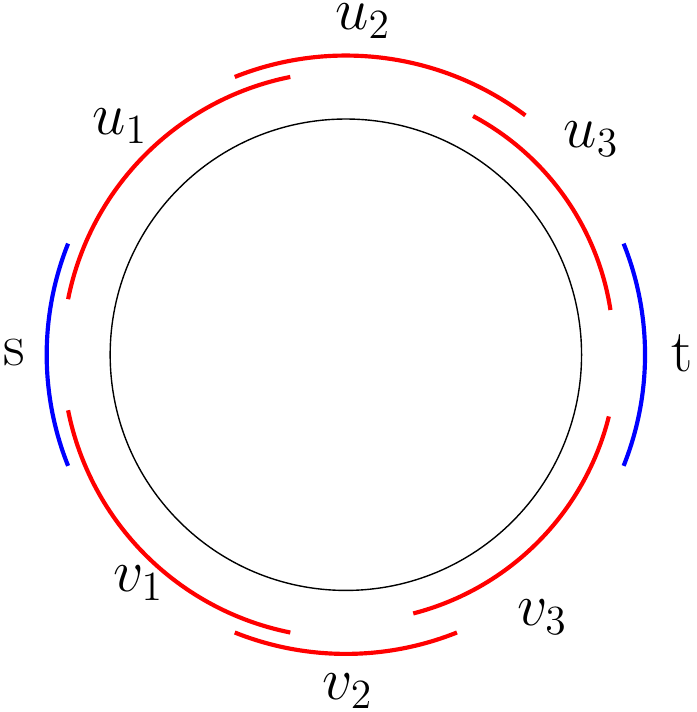}
    \vspace{0.5cm}
    \caption{A circular-arc graph with $d(s,t)=4$}
    \label{fig:circulararc}
\end{figure}

Given a graph, its circular-arc representation can be constructed in linear time, if it exists~\cite{tucker,kaplan}.

\begin{theorem}
$\spr$ can be solved in polynomial time for circular-arc graphs.
\end{theorem}

\begin{proof} Let $(G,s,t,P,Q)$ be an $\spr$ instance, where $G$ is a circular-arc graph. If $d(s,t) \leq 5$, then we can solve $\spr$ in polynomial time (\autoref{thm:constdia}). If $d(s,t) \geq 6$, then each $s$--$t$ shortest path has at least 7 vertices. Note that the arcs corresponding to the middle vertices (at distance $\lfloor d(s,t)/2 \rfloor$ from $s$) of $P_1$ and $P_2$ (say $v$ and $u$) do not intersect any arc that intersects $s$ or $t$. Also it is easy to see that the removal of the arcs $s$ and $t$ from the circle (denoted by $G\setminus\{s,t\}$) divides the circle into two disjoint arcs.
\begin{itemize}
    \item Case 1: If the arcs $v$ and $u$ lie on the same arc of $G\setminus\{s,t\}$, then we can think of the arcs as intervals of an interval graph. And since interval graphs are chordal, we know from~\cite{Bonsma13} that $\spr$ is polynomial-time solvable for them.
    
    \item Case 2: If the arcs $v$ and $u$ lie on different arcs of $G\setminus\{s,t\}$, then $P_1$ and $P_2$ cannot be reconfigured. In particular, $v$ can never be reconfigured to $u$ because the two neighbours of $v$ on $P_1$ are from one arc of $G\setminus\{s,t\}$, and the two neighbours of $u$ on $P_2$ are from the other arc of $G\setminus\{s,t\}$.
\end{itemize}
This completes the proof.
\end{proof}

\subsection{\texorpdfstring{$\spr$}{SPR} for Graphs of Constant Diameter}
\label{sec:bdd-diam}
\begin{theorem} \label{thm:constdia}
Let $G$ be an $n$-vertex graph such that $d(s,t)=c$. Then $\spr$ can be solved in $n^{O(c)}$ time for $G$.
\end{theorem}

\begin{proof}
Note that the $\bfs$ tree from $s$ to $t$ has $c$ layers, each layer with at most $n$ vertices. Thus, the number of $s$--$t$ shortest paths is at most $n^c$. This means $\spr(G,s,t)$ has at most $n^c$ vertices and therefore at most $n^{2c}$ edges. Hence, given $G$, $\spr(G,s,t)$ can be constructed and connectivity of two vertices in $\spr(G,s,t)$ can be checked in $n^{O(c)}$ time.
\end{proof}

\begin{definition}\
\begin{enumerate}
    \item A~\textbf{bipartite} graph is a graph whose vertex set can be partitioned into two independent sets.
    \item A~\textbf{split} graph is a graph whose vertex set can be partitioned into two sets: an independent set and a clique.
    \item A~\textbf{co-bipartite} graph is a graph whose vertex set can be partitioned into two cliques.
\end{enumerate}
\end{definition}

It is noteworthy that $\spr$ behaves differently on these three closely related graph classes. By~\autoref{obs:bipart}, $\spr$ is $\PSPACE$-complete for bipartite graphs.

\begin{proof}[Proof of~\autoref{obs:bipart}] Let $(G,s,t,P,Q)$ be an $\spr$ instance. Consider the layered BFS tree from $s$ to $t$, and delete all intra-layer edges (edges connecting two vertices in the same layer) from it, since these edges cannot be on an $s$--$t$ shortest path. The resulting graph $G'$ is bipartite: odd layer vertices form one partition, and even layer vertices form the other. $P$ and $Q$ are reconfigurable in $G'$ if and only if they are reconfigurable in $G$. This completes the proof.
\end{proof}
In contrast, $\spr$ is solvable in polynomial time for split graphs and co-bipartite graphs. The following fact about split graphs and co-bipartite graphs is well-known and easy to see~\cite{fact1}.
\begin{observation}
The graph diameter of split graphs and co-bipartite graphs is at most 3.
\end{observation}

This implies $d(s,t)\leq 3$ for split graphs and co-bipartite graphs, leading to the following corollary of~\autoref{thm:constdia}.

\begin{corollary}
$\spr$ can be solved in polynomial time for split graphs and co-bipartite graphs.
\end{corollary}

\section{Gradation of the Complexity of \texorpdfstring{$\kspr$}{k-SPR}} \label{sec:gradationcomplexitykspr}

In this section, we will see that for a fixed $n$, the complexity of $\kspr$ can~\emph{decrease} as $k$ increases, varying from $1$ to $n$. Note that changing $k$ contiguous vertices requires a cycle of size $2k + 2$. We refer to such a cycle (which changes $k$ contiguous vertices) as a $k$-switch.


\begin{lemma}[Monotonicity of Diameter] Let $k_1$ and $k_2$ be two positive integers such that $k_1\leq k_2$. Then the graph diameter of $k_1$-$\spr(G,s,t)$ is at least the graph diameter of $k_2$-$\spr(G,s,t)$.
\end{lemma}

\begin{proof}
The vertex sets of the graphs $k_1$-$\spr(G,s,t)$ and $k_2$-$\spr(G,s,t)$ are the same (one vertex for each $s$--$t$ shortest path in $G$). Since every edge of $k_1$-$\spr(G,s,t)$ is also present in $k_2$-$\spr(G,s,t)$, this completes the proof.
\end{proof}

\begin{theorem}
For every integer $k\geq 1$, there exists an $n$-vertex graph $G$ such that the diameter of $\kspr(G,s,t)$ is $2^{\Omega(n/k)}$.
\end{theorem}

\begin{proof}
\cite{KaminskiMM11} showed that there exists an $n$-vertex graph $G$ for which the diameter of $\spr(G,s,t)$ is $2^{\Omega\left(n\right)}$. Our graph is a simple modification of theirs.

Consider any odd\footnote{It is easy to see that our proof also works for even $k$.} $k$ (thus $k=2\ell + 1$ for some $\ell$). Subdivide each edge $\ell$ times (equivalently, replace each edge by a path $P_{\ell+2}$ between its end points). Call this new graph $G_\ell$. We make the following claim.

\begin{claim} \label{cl:switch} For all $1\leq k'\leq k-1$, there is no $k'$-switch to reconfigure two $s$--$t$ shortest paths in $G_\ell$.
\end{claim}

Note that every $1$-switch in $G$ was directly converted to a $k$-switch in $G_\ell$. Further, the start and end vertices of a $k'$-switch (say $u$ and $v$) have at least two neighbours in the next or previous layer, while all the newly added vertices in $G_\ell$ have only one neighbour in the next layer and only one neighbour in the previous layer. Thus, $u$ and $v$ cannot be newly added vertices in $G_\ell$, implying that they are vertices of the original graph $G$.

Finally, we will show that $d_G(u,v)=2$. If $d_G(u,v)\geq 3$, then $d_{G_\ell}(u,v)\geq 3\ell+3$ (as every two vertices of $G$ have distance at least $\ell+1$ in $G_\ell$). And if $d_G(u,v)\leq 1$, then $u$ and $v$ are adjacent in $G$, and the existence of two paths between them in $G_\ell$ implies that there is a multiple edge $(u,v)$ in $G$, a contradiction since our graphs are simple. This completes the proof of~\autoref{cl:switch}.

\autoref{cl:switch} implies that every switch in $G_\ell$ is a $k$-switch which can be mapped back to a 1-switch in $G$. Consider the two paths $P_1$ and $P_2$ in $G$ which require $2^{\Omega(n)}$ 1-switches. These correspond to two paths $P_1'$ and $P_2'$ which require $2^{\Omega(n)}$ $k$-switches in $G_\ell$.

It is easy to see that the degree of each vertex in the graph $G$ from~\cite{KaminskiMM11} is upper-bounded by a constant. Thus, $G$ has $O(n)$ edges. Since each edge of $G$ is subdivided $\ell$ times in $G_\ell$, the number of vertices in $G_\ell$ is $N=O(n\ell)=O(nk)$. Hence, the diameter of $\kspr(G_\ell,s,t)$ is at least $2^{\Omega(n)}=2^{\Omega(N/k)}$.
\end{proof}

\autoref{fig:complexity} illustrates known complexity results for $\kspr$ as $k$ varies from 1 to $n$, for a fixed $n$. \autoref{lem:ksprlinegraphs} implies that $\kspr$ is $\PSPACE$-complete when $k=O(1)$.

\begin{theorem}
$\kspr$ can be solved in polynomial time when $k \geq n/2$.
\end{theorem}

\begin{proof} We can split the proof into two cases.

\textbf{Case 1: $d(s,t)< n/2$.} This means $d(s,t)< n/2\leq k$. Since we are allowed to reconfigure up to $k$ contiguous vertices in one reconfiguration step of $\kspr$, we can trivially reconfigure $P_{1}$ to $P_{2}$ in polynomial time.

\textbf{Case 2: $d(s,t)\geq n/2$.} This means at least $n/2$ vertices of the graph lie on $P_1$. So $P_2$ contains at most $n/2$ vertices that do not lie on $P_1$. Since $k\geq n/2$, we can trivially reconfigure $P_{1}$ to $P_{2}$ in polynomial time.
\end{proof}

\begin{figure}
    \centering
    \vspace{0.5cm}
    \includegraphics[width=0.6\linewidth]{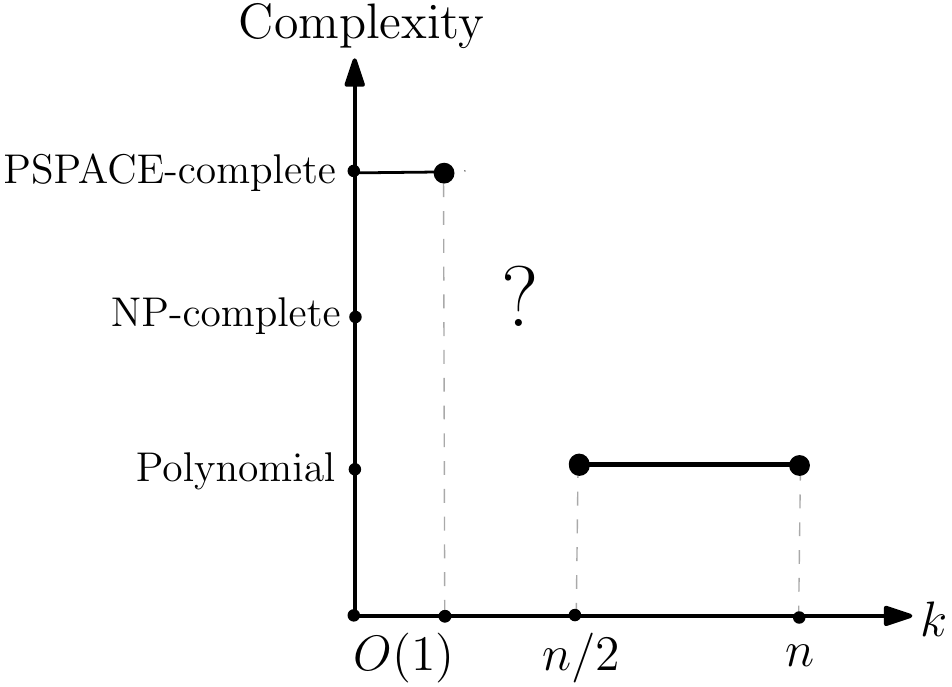}
    \vspace{0.5cm}
    \caption{Complexity of $\kspr$ as $k$ varies from 1 to $n$}
    \label{fig:complexity}
\end{figure}

\section{Optimization Variants of~\texorpdfstring{$\spr$}{SPR}} \label{sec:optvar}
In this section, we define three variants of the Shortest Path Reconfiguration problem. In this settings, we are allowed to change any number of vertices at a time. In addition, we pay a price of $p_{i}$ for changing $i$ vertices on a path. Furthermore, $$p_{1} \leq p_{2} \leq \cdots \leq p_{n - 1} \leq p_{n}.$$

\begin{definition}[MinSumSPR]
Given $(G,s,t,P_1,P_2)$, an $\spr$ instance, output a reconfiguration sequence from $P_1$ to $P_2$ (if it exists) that minimises the~\textbf{total} cost of reconfiguration. 
\end{definition}

\begin{definition}[MinMaxSPR]
Given $(G,s,t,P_1,P_2)$, an $\spr$ instance, output a reconfiguration sequence from $P_1$ to $P_2$ (if it exists) that minimises the~\textbf{maximum} cost of reconfiguration.  
\end{definition}

Generalizing MinSumSPR and MinMaxSPR, we get the following.

\begin{definition}[MinTop-$\ell$-$\spr$]
Given $(G,s,t,P_1,P_2)$, an $\spr$ instance, output a reconfiguration sequence from $P_1$ to $P_2$ (if it exists) that minimises the~\textbf{sum total of the maximum $\ell$ (or top-$\ell$)} costs of reconfiguration.
\end{definition}

Note that MinSumSPR is a special case of MinTop-$\ell$-$\spr$ with $\ell=\infty$ and MinMaxSPR is a special case of MinTop-$\ell$-$\spr$ with $\ell=1$.

\begin{lemma} \label{lem:graphdiam}
For every $n$-vertex graph $G$, the diameter of $\kspr(G,s,t)$ is at most $2^{n}$.
\end{lemma}

\begin{proof}
Each $s$--$t$ shortest path in $\kspr(G,s,t)$ is a~\emph{distinct} subset of vertices of $G$. As $G$ has $n$ vertices, the $\kspr$ graph has at most $2^{n}$ vertices.
\end{proof}

\begin{theorem}
There is no polynomial-time algorithm that approximates MinTop-$\ell$-$\spr$ to within a factor of $O(2^{n^2})$, unless $\PSPACE = \P$.
\end{theorem}

\begin{proof}
For the sake of contradiction, assume that there is an $O(2^{n^{2}})$-factor approximation algorithm for MinTop-$\ell$-$\spr$. We reduce the original $\spr$ problem, which is known to be $\PSPACE$-complete~\cite{Bonsma13}, to MinTop-$\ell$-$\spr$. Let $I = (G = (V, E), s, t, P_{1}, P_{2})$ be an instance of $\spr$.

We assign the following costs for changing $i$ contiguous vertices, denoted by $p_{i}$.
\begin{equation}
  p_{i} =
    \begin{cases}
      1, & i = 1;\\
      \ell\cdot2^{n^{2} + 1}, & i \geq 2.
    \end{cases}
\end{equation}
If $I \in \spr$, there exists a reconfiguration sequence which changes one vertex per reconfiguration step. The number of such steps is at most $\diam(\spr(G,s,t))$, thus $\opt$ pays at most $\ell$ since all moves are changing $1$ vertex, and cost is $1$.

Thus, the cost paid by $\opt$ is at most $2^{n}$. Since $\alg$ is an $O(2^{n^{2}})$-approximation algorithm, $\alg\leq\ell\cdot2^{n^{2}}$.

Conversely, assume that $I \notin \spr$. Then every reconfiguration sequence must change at least $2$ vertices in at least in $1$ reconfiguration step. Thus, in the MinTop-$\ell$-$\spr$ instance, we pay a cost of $\ell\cdot 2^{n^{2} + 1}$ at least once. As $\opt$ is lower bounded by this quantity, $\alg\geq\ell\cdot 2^{n^{2} + 1}$.

\autoref{algo:reduc} ($\reduc$) explains the reduction formally. The correctness of~\autoref{algo:reduc} is shown by the above statements. The reduction is polynomial time, because writing $2^{n^{2} + 1}\cdot \ell \leq 2^{n^{2} + 2n + 1}$ takes at most $\poly(n)$ bits, and $\alg$ runs in polynomial time. Thus, if there is an $O(2^{n^{2}})$-approximation for MinTop-$\ell$-$\spr$, then we can solve $\spr$ in polynomial time, implying that $\PSPACE=\P$.
\end{proof}

\begin{algorithm}
\caption{$\reduc$}
\label{algo:reduc}
\textbf{Input:} An instance of $\spr$
\begin{algorithmic}[1]
\State Creates MinTop-$\ell$ $\spr$, by setting costs as described
\State Run $\alg$ on this instance
\If{$c(\alg) \leq 2^{n^{2} + 1}\cdot l$}
\State \textbf{Output} $I \in \spr$
\Else
\State \textbf{Output} $I \notin \spr$
\EndIf 
\end{algorithmic}
\end{algorithm}

\begin{corollary}
There is no polynomial-time algorithm that approximates MinSumSPR to within a factor of $O(2^{n^2})$, unless $\PSPACE = \P$.
\end{corollary}

\begin{corollary}
There is no polynomial-time algorithm that approximates MinMaxSPR to within a factor of $O(2^{n^2})$, unless $\PSPACE = \P$.
\end{corollary}

\section{Number of Shortest Paths versus Length of Shortest Path} \label{sec:tradeoff}

In this section, we study how the number of vertices in $\spr(G,s,t)$ varies with the path length between $s$ and $t$. We denote $|V_{\spr}|$ as a function $f$ that takes $d(s, t)$ as input.

\begin{figure}[ht]
    \centering
    \vspace{0.5cm}
    \includegraphics[width=0.7\linewidth]{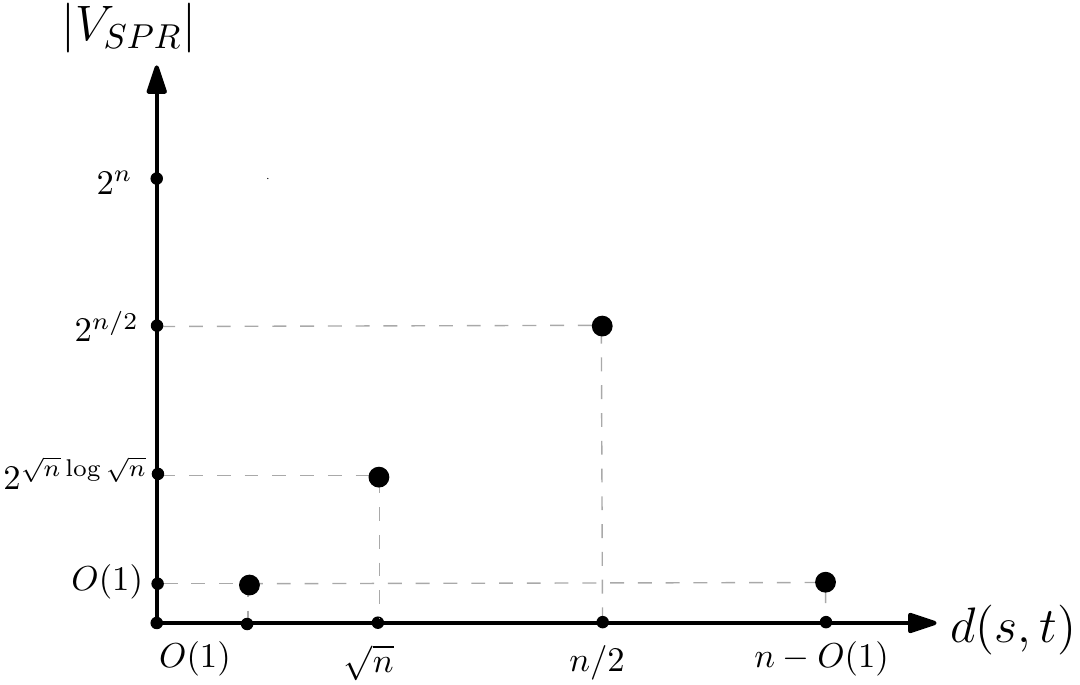}
    \vspace{0.5cm}
    \caption{Size of $\spr(G,s,t)$ as $d(s,t)$ varies from 1 to $n$}
    \label{fig:dstvsvspr}
\end{figure}

It is easy to see that $f(x) = O(2^{n})$ (\autoref{lem:graphdiam}) for all values of $d(s,t)$. For specific values of $d(s,t)$, we have the following stronger bounds, represented by~\autoref{fig:dstvsvspr}.

\begin{lemma}
\begin{align*}
d(s, t) &= x = n/2 &&\Rightarrow\quad f(x) = \Theta(2^{n});\\ d(s, t) &= x = \Theta(\sqrt{n}) &&\Rightarrow\quad f(x) = \Omega\left(2^{\sqrt{n}\log \sqrt{n}}\right);\\
d(s, t) &= x = n - O(1) &&\Rightarrow\quad f(x) = O(1).
\end{align*}
\end{lemma}

\begin{proof}
Consider the following gadget graph on $l+2$ vertices. Two independent vertices $u$ (called the start point of the gadget) and $v$ (called the end point of the gadget) are both connected to a set of $l$ independent vertices. Identifying the end point of the gadget with the start point of a second gadget adds $l + 1$ more vertices. If we chain $g$ gadgets in series in this way, the resulting graph has $1 + g(l + 1)$ vertices.

Let $s$ and $t$ be the start point of the first gadget and the end point of the last gadget, respectively. The graph has $n=1+g(l+1)$ vertices, $d(s,t)=2g=\Theta(\frac{n}{l})$, and the number of $s$--$t$ shortest paths is $l^{g}$. When $l = c$ for some constant $c$, we get $c^{O(\frac{n}{c})} = \Theta(2^{n})$. Similarly, for $l = \Theta(\sqrt{n}), g = \Theta(\sqrt{n})$, we get path length of $\Theta(\sqrt{n})$ and roughly $\sqrt{n}^{\sqrt{n}} = \Theta(2^{\sqrt{n}\log n})$ paths. This shows the first two implications.

Let $d(s, t) = n - c$, for some constant $c$. These $c$ vertices are present in at most $c$~\emph{distinct} layers, and each layer has at most $c$ vertices. Thus, the number of possible reconfigurations is at most $c^{c}$, which is a constant. This shows the third implication. 
\end{proof}

\section{Conclusion}

We conclude with some possible directions for future research on $\spr$ and $\kspr$.

\begin{itemize}

    \item \textbf{Experimental results:} This entire paper contains only theoretical results. One can try to simulate or observe the behaviour of $\spr$ (or $\kspr$) on practical instances, and see if they yield something more substantial or meaningful than what we could obtain theoretically.
    
    \item \textbf{A graph invariant of $\spr$:} Since $\spr$ is $\PSPACE$-hard in general, we study $\spr$ on specific graph classes and provide polynomial-time algorithms for some of them, each one requiring its own separate proof. Is there a graph invariant (like treewidth) that characterizes the complexity of $\spr$?
    
    \item \textbf{$\spr$ for unit disk graphs:} A highly practical application of $\spr$ is in rerouting messages across a network of telecommunication towers. Each tower has a coverage radius (known as range), and two towers can interact if one lies in the range of the other. When one tower fails (this could happen, for instance, due to a high rate of incoming messages at the tower, leading to an overload), we need to quickly reroute messages using a different set of towers, which is precisely the $\spr$ problem for unit disk graphs. This is seemingly a very difficult problem, and even an approximation algorithm for it will be very useful.
    
    \item \textbf{Parameterized complexity:} Although both $\spr$ and $\kspr$ are $\PSPACE$-complete, they might be fixed-parameter tractable, for the right choice of the parameter. It will be interesting to explore which parameters could work (if any).
    
    \item \textbf{A novel variant of reconfigurability:} What if shortest paths need not change in $k$ contiguous vertices at a time? A variant of reconfigurability could be defined as follows: One shortest path can ``hop'' on to another shortest path, whenever it finds that the $(i+1)$-th vertex of the other path is adjacent to its $i$-th vertex. This notion of reconfigurability is also practical, as it allows shortest paths to directly use other shortest paths that already exist.
    
\end{itemize}

\bibliography{main_long_arxiv_version}

\end{document}